\documentclass[a4paper,11pt]{article}
\pdfoutput=1
\usepackage{jheppub,silence}
\usepackage[ascii]{inputenc}
\usepackage[dvipsnames,svgnames]{xcolor}
\usepackage{bbm,braket,microtype,mathrsfs,amsmath,amssymb,amsthm,amsfonts,latexsym,mathtools,graphicx,enumitem,booktabs,bm,float,mathdots,caption,subcaption,ellipsis,enumitem,overpic,mleftright}
\usepackage[T1]{fontenc}
\usepackage[english]{babel}
\usepackage[capitalize,nameinlink]{cleveref}
\newtheorem{thm}{Theorem}[section]\crefname{thm}{Theorem}{Theorems}
\newtheorem{lem}[thm]{Lemma}\crefname{lem}{Lemma}{Lemmas}
\crefname{prop}{Proposition}{Propositions}
\numberwithin{equation}{section}

\DeclareMathOperator{\tr}{tr}

\DeclareMathOperator{\GHZ}{GHZ}
\DeclareMathOperator{\bulk}{bulk}
\DeclareMathOperator{\area}{Area}

\DeclarePairedDelimiter\abs{\lvert}{\rvert}
\DeclarePairedDelimiter\norm{\lVert}{\rVert}
\DeclarePairedDelimiter\parens{\lparen}{\rparen}
\DeclarePairedDelimiter\braces{\lbrace}{\rbrace}

\newcommand{\ot}{\otimes}
\newcommand{\id}{\mathrm{id}}
\newcommand{\CC}{\mathbbm{C}}

\newcommand{\ZZ}{\mathbbm{Z}}
\newcommand{\EE}{\mathbbm{E}}
\newcommand{\HH}{\mathcal{H}}
\newcommand{\eps}{\varepsilon}
\newcommand{\bigO}{\mathcal O}

\newcommand{\proj}[1]{\mathinner{\lvert#1\rangle\!\langle#1\rvert}}
\newcommand{\ketbra}[2]{\mathinner{\lvert#1\rangle\!\langle#2\rvert}}

\begin{document}

\title{Fun with replicas: \texorpdfstring{\\}{} tripartitions in tensor networks and gravity}
\hypersetup{pdftitle={Fun with replicas: tripartitions in tensor networks and gravity}, pdfauthor={Geoff Penington, Michael Walter, Freek Witteveen}}
\date{}
\author[1,2]{Geoff Penington}
\emailAdd{geoffp@berkeley.edu}
\author[3]{Michael Walter}
\emailAdd{michael.walter@rub.de}
\author[4]{Freek Witteveen}
\emailAdd{fw@math.ku.dk}
\affiliation[1]{Center for Theoretical Physics and Department of Physics, University of California,~Berkeley,~USA}
\affiliation[2]{Institute for Advanced Study, Princeton, USA}
\affiliation[3]{Faculty of Computer Science, Ruhr University Bochum, Germany}
\affiliation[4]{Department of Mathematical Sciences and QMATH, University of Copenhagen, Denmark}
\abstract{We analyse a simple correlation measure for tripartite pure states that we call~$G(A:B:C)$.
The quantity is symmetric with respect to the subsystems~$A$,~$B$,~$C$, invariant under local unitaries, and is bounded from above by $\log d_A d_B$.
For random tensor network states, we prove that $G(A:B:C)$ is equal to the size of the \emph{minimal tripartition} of the tensor network, i.e., the logarithmic bond dimension of the smallest cut that partitions the network into three components with $A$, $B$, and~$C$.
We argue that for holographic states with a fixed spatial geometry, $G(A:B:C)$ is similarly computed by the minimal area tripartition.
For general holographic states, $G(A:B:C)$ is determined by the minimal area tripartition in a backreacted geometry, but a \emph{smoothed} version is equal to the minimal tripartition in an unbackreacted geometry at leading order.
We briefly discuss a natural family of quantities~$G_n(A:B:C)$ for integer~$n \geq 2$ that generalize $G=G_2$.
In holography, the computation of $G_n(A:B:C)$ for $n>2$ spontaneously breaks part of a $\ZZ_n \times \ZZ_n$ replica symmetry.
This prevents any naive application of the Lewkowycz-Maldacena trick in a hypothetical analytic continuation to~$n=1$.}
\maketitle

\section{Introduction}
The Ryu-Takayanagi formula~\cite{ryu2006aspects} states that the entanglement entropies in holographic CFTs are related via the AdS/CFT dictionary to the area of minimal surfaces in the dual bulk geometry.
Specifically, for static states, at the leading classical order in the Newton constant~$G_N$, we have%
\footnote{This formula has since been generalized to time-dependent bulk geometries and to include quantum corrections~\cite{hubeny2007covariant,faulkner2013quantum,engelhardt2015quantum}.}
\begin{align} \label{eq:RT}
    S(A) = \min_{\gamma_A \cup A = \partial a} \frac{\area(\gamma_A)}{4G_N},
\end{align}
where $A$ is any boundary subregion and $\area(\gamma_A)$ is the area of the surface~$\gamma_A$ found by minimizing over all surfaces homologous to $A$ within a static slice of the bulk spacetime.
A closely analogous formula exists for random tensor networks~\cite{hayden2016holographic} with maximally entangled links, where the entanglement entropy of a set of boundary legs $A$ is at leading order in the bond dimension $D$ given by
\begin{align} \label{eq:TNRT}
    S(A) = \min_{\gamma_A} \, \abs{\gamma_A} \log D = \min_{\gamma_A} \log d_{\gamma_A},
\end{align}
where the minimization is now over all cuts~$\gamma_A$ through the network homologous to~$A$,
and~$d_{\gamma_A}$ is the product of the dimensions of all legs in the cut $\gamma_A$.
This analogy has led to considerable interest in tensor networks as a toy model of quantum gravity~\cite{hayden2016holographic,yang2016bidirectional,qi2017holographic,qi2018spacetime,penington2022replica,nezami2020multipartite,dong2021holographic,dong2022replica,akers2022reflected,dutta2021canonical,kudler2022negativity,qi2022holevo,apel2022holographic,cheng2022random,akers2022reflected2}.

The standard derivation of both \cref{eq:RT} and \cref{eq:TNRT} uses the \emph{replica trick}.
One first computes integer $n$ R\'{e}nyi entropies for the density matrix $\rho_A = \tr_{\bar A}[\ketbra{\Psi}{\Psi}]$
\begin{align} \label{eq:Renyireplica}
    S_n(A) = \frac{1}{1-n} \log\tr[\rho_A^n] = \frac{1}{1-n} \log \bra{\Psi}^{\otimes n} \tau_A \ket{\Psi}^{\otimes n}
\end{align}
by writing them in terms of the expectation value of an operator $\tau_A$ acting on $n$ copies of the state $\ket{\Psi}$.
Specifically, $\tau_A$ acts by permuting the $n$ copies of the subsystem $A$ in a cyclic permutation $\tau$.
(In general, given a subsystem $A$ and a permutation $\pi \in S_n$ we will denote by $\pi_A$ the operator which applies the permutation $\pi$ to the $n$ copies of the system $A$.) 
The entanglement entropy can be computed from \cref{eq:Renyireplica} by analytically continuing to $n=1$.

For random tensor networks, there is in fact no need to do an analytic continuation in order to obtain an entropy equal to $\log d_{\gamma_A}$, with $\gamma_A$ the minimal cut.
In the limit of large bond dimensions, all the R\'{e}nyi entropies $S_n$ are the same at leading order, and so $S(A)$ can be replaced by $S_n(A)$ without altering \cref{eq:TNRT}.
In AdS/CFT, this is not the case; for typical semiclassical states, such as the CFT ground state, it is just the entanglement entropy $S(A)$ which is holographically dual to the minimal surface area $\area(\gamma_A)/4G_N$, while the R\'{e}nyi entropies $S_n(A)$ are related to areas in backreacted geometries~\cite{lewkowycz2013generalized,dong2016gravity}.
However, the \emph{smoothed R\'{e}nyi entropies}
\begin{align*}
    S_n^\eps(\rho) = \max_{\tilde\rho \, \overset{\eps}{\approx} \, \rho} S_n(\tilde\rho),
\end{align*}
where the maximization is over states $\tilde\rho$ that are $\eps$-close to the state~$\rho$ (say in trace distance) are equal to $\area(\gamma_A)/4G_N$ up to subleading $\bigO(G_N^{-1/2})$ corrections~\cite{bao2019beyond}.

In this paper we discuss generalizations of the replica trick that involve different permutations being applied to multiple bulk subregions.
Such generalizations have been used to compute multipartite correlation measures such as entanglement negativities, reflected entropies and the realignment criteria~\cite{aubrun2012partial,aubrun2012realigning,nezami2020multipartite,dong2021holographic,dong2022replica,dutta2021canonical,akers2022reflected,kudler2022negativity}.%
\footnote{In fact, the replica trick used to compute $G(A:B:C)$ has previously appeared as part of analytic continuations used to compute the realignment criterion and the reflected entropy.}
Here, we focus on a particularly symmetric measure $G(A:B:C)$ that is defined for any  pure state $\ket{\Psi}_{ABC}$ in a Hilbert space made up of three subsystems $A$, $B$, and $C$, and that is computed using four replicas of~$\ket{\Psi}$.
Specifically, we define
\begin{align}\label{eq:G def}
    G(A : B : C) = -\frac{1}{2} \log \bra{\Psi}^{\otimes 4} \pi^{(1 2)(3 4)}_A \pi^{(1 3)(2 4)}_B \pi^{(1 4)(2 3)}_C \ket{\Psi}^{\otimes 4}.
\end{align}
Here, e.g., $\pi^{(1 2) (3 4)}_A$ acts on subsystem $A$ and swaps replica 1 with replica 2 and replica 3 with replica 4.
It is clear from this definition that $G(A:B:C)$ is invariant under permuting or relabeling the subsystems $A$, $B$, and $C$.
Since $G(A:B:C)$ is also invariant under local unitaries acting on any of the three subsystems, it is in fact uniquely determined by any of three reduced density matrices $\rho_{AB}$, $\rho_{BC}$, or $\rho_{AC}$ on a pair of subsystems.
However the full $S_3$~symmetry under permutations of $A$, $B$, and $C$ is only manifest when described using the pure tripartite state $\ket{\Psi}$.

In \cref{sec:multipartite entanglement replica trick}, we review the replica trick and its generalization to multiple subsystems.
We then prove a number of important properties for $G(A:B:C)$, including that
\begin{align}\label{eq:bounds on G}
    0 \leq G(A : B : C) \leq \min\,\braces[\big]{ \log d_A + \log d_B, \log d_A + \log d_C, \log d_B + \log d_C },
\end{align}
where $d_A, d_B$, and $d_C$ are the dimensions of the Hilbert spaces of systems $A$, $B$, and~$C$.
We also show that in the special case of purely bipartite entangled states $\ket{\Psi}_{ABC} = \ket{\psi_1}_{A_1 B_1} \ot \ket{\psi_2}_{A_2 C_1} \ot \ket{\psi_3}_{B_2 C_2}$ (with $\mathcal{H}_A \cong \mathcal{H}_{A_1} \otimes \mathcal{H}_{A_2}$ and so on) we have
\begin{align*}
    G(A : B : C) = \frac12 \parens[\big]{ S_2(A) + S_2(B) + S_2(C) }.
\end{align*}

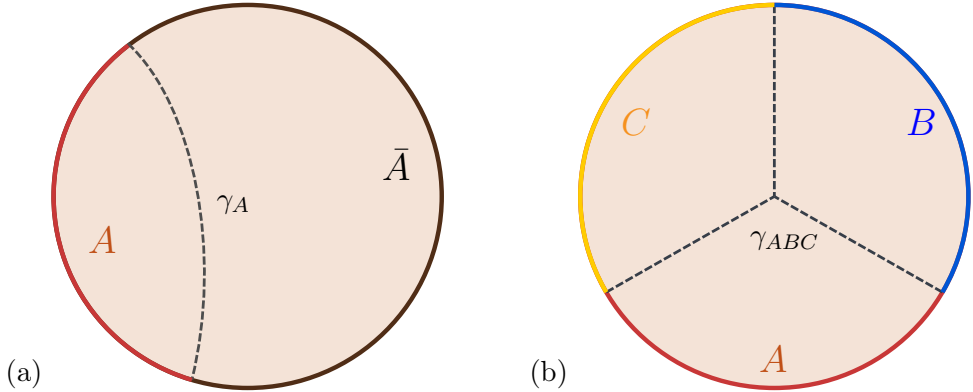
\begin{figure}
\begin{center}
\begin{overpic}[width=.8\linewidth,grid=false]{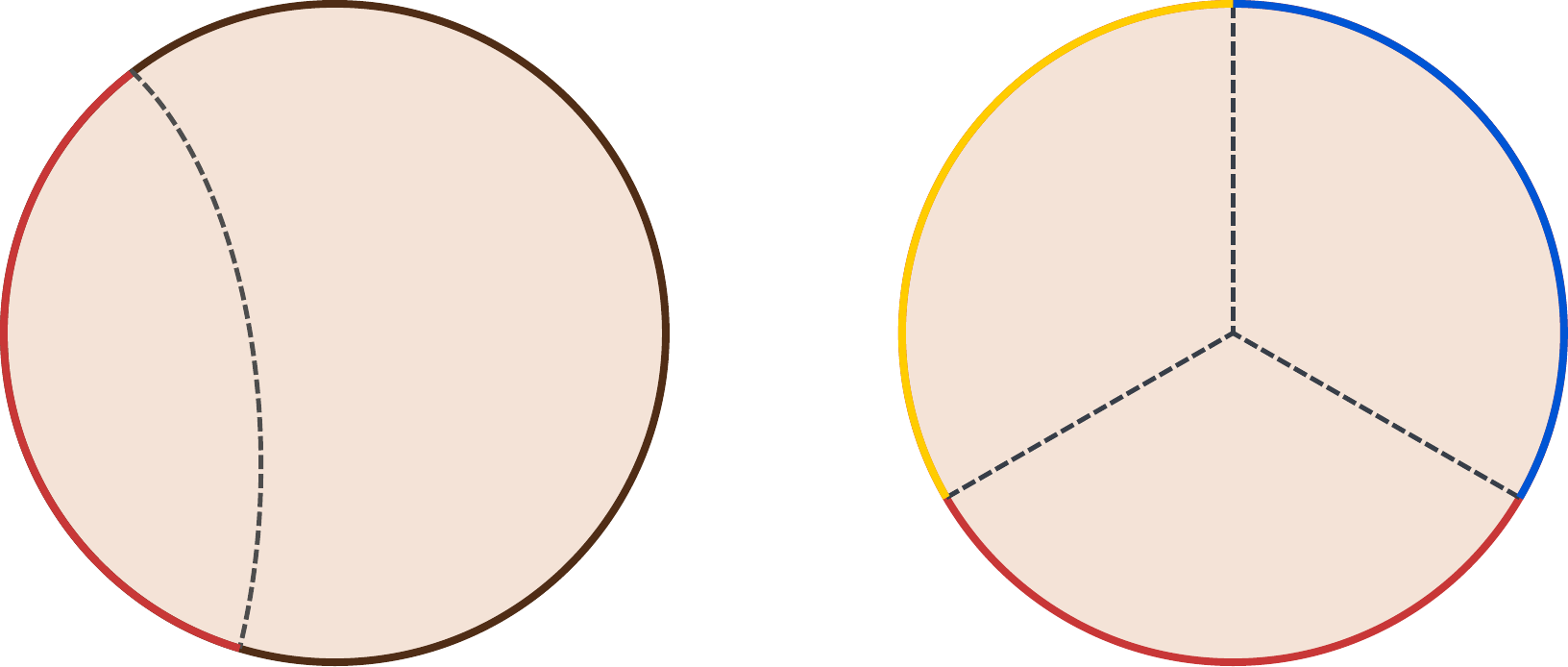}
    \put(4,15){\color{Bittersweet}{\Large{$A$}}}
    \put(36,23){\Large{$\bar{A}$}}
    \put(18,20){$\gamma_A$}
    \put(77,2){\color{Bittersweet}{\Large{$A$}}}
    \put(93,28){\color{Blue}{\Large{$B$}}}
    \put(62,28){\color{BurntOrange}{\Large{$C$}}}
    \put(76,16){$\gamma_{ABC}$}
    \put(-5,1){(a)} \put(52,1){(b)}
\end{overpic}
\end{center}
\caption{(a)~The von Neumann entropy of a boundary region $A$ for either a random tensor network or a holographic CFT state is computed by the area of the minimal surface $\gamma_A$ separating $A$ from its complement $\bar{A}$.
(b)~In this work we argue that in random tensor networks and fixed area states the quantity $G(A : B : C)$ defined in \cref{eq:G def} is computed by a minimal tripartition~$\gamma_{ABC}$, and that upon smoothing this is also approximately valid for general holographic states.\label{fig:mercedes star}}
\end{figure}

In \cref{sec:rtn}, we compute $G(A:B:C)$ for states created using random tensor networks.
In the limit of large bond dimension, we prove that $G(A:B:C)$ is given by the size of the \emph{minimal tripartition} of the tensor network separating $A$, $B$, and $C$.
Let us explain what we mean by a minimal tripartition (also know as a 3-terminal cut or 3-way cut \cite{dahlhaus1994complexity,vazirani2001approximation}).
The tensor network graph consists of bulk vertices and boundary vertices, where the boundary vertices are partitioned into subsets~$A$, $B$, and $C$.
A tripartition can be described as a set~$\gamma_{ABC}$ of edges such that upon removing the edges in $\gamma_{ABC}$ from the original graph, $A$, $B$, and $C$ are no longer pairwise connected.%
\footnote{Formally, we define a tripartition to be a partitioning $V = \Gamma_A \cup \Gamma_B \cup \Gamma_C$ of the vertices in three disjoint sets such that $A \subseteq \Gamma_A$, $B \subseteq \Gamma_B$ and $C \subseteq \Gamma_C$.
The set $\gamma_{ABC}$ is then the set of edges between different sets.}
A minimal tripartition is a tripartition of minimal size.
We prove that if the minimal tripartition $\gamma_{ABC}$ is unique, for large bond dimension $D$,
\begin{align*}
    G(A : B : C) \approx \min_{\gamma_{ABC}} \, \abs{\gamma_{ABC}} \log D = \min_{\gamma_{ABC}} \log(d_{\gamma_{ABC}})
\end{align*}
where $\abs{\gamma_{ABC}}$ is the size of the tripartition and $d_{\gamma_{ABC}}$ is the total bond dimension along it.
In contrast to the usual minimal cuts in graphs which can be found in polynomial time, interestingly, minimal 3-terminal cuts are NP-hard to compute~\cite{dahlhaus1994complexity}.
Note that if $\gamma_{A}$ and $\gamma_B$ are minimal surfaces for $A$ and $B$ respectively, then $\gamma_A \cup \gamma_B$ is a tripartition for $A$, $B$, and~$C$, but not necessarily a minimal one.
This is consistent with the upper bound in \cref{eq:bounds on G}.

In \cref{sec:gravity}, we turn out attention to holographic boundary states in AdS/CFT.
For a certain class of states, where the geometry on the static spatial slice is approximately fixed, we show that at leading order in the Newton constant~$G_N$,
\begin{align*}
    G(A:B:C) = \min_{\gamma_{ABC}} \frac{\mathrm{Area}(\gamma_{ABC})}{4G_N}
\end{align*}
where the minimization is similarly over (piecewise smooth) surfaces that tripartition the bulk geometry.
For example, if $A$, $B$, and $C$ are contiguous intervals at the boundary of vacuum AdS$_3$, then the surface~$\gamma_{ABC}$ is the ``Mercedes star'' shown in \cref{fig:mercedes star}.

For general holographic states, $G(A:B:C)$ is determined by the minimal area tripartition in a \emph{backreacted} geometry, analogous to the backreacted geometries that appear in R\'{e}nyi entropy computations.
However, again in close analogy to R\'{e}nyi entropies, one can define a smoothed quantity
\begin{align*}
    G^{\eps}(A : B : C)_{\ket{\Psi}} = \max_{\ket{\widetilde\Psi}\approx_\eps \ket{\Psi}} G^{\eps}(A : B : C)_{\ket{\widetilde\Psi}},
\end{align*}
which satisfies
\begin{align*}
     G^{\eps}(A : B : C)_{\ket{\Psi}} = \min_{\gamma_{ABC}} \frac{\area(\gamma_{ABC})}{4G_N} + \bigO\mleft(G_N^{-1/2}\mright).
\end{align*}
Here the minimization is over tripartitions $\gamma_{ABC}$ in the original, \emph{unbackreacted} geometry, as for fixed geometry states.

Finally in \cref{sec:generalize} we discuss possible generalizations of $G(A:B:C)$.
We define a family of quantities $G_n(A:B:C)$ for $n \geq 2$ that involve $n^2$ replicas, with $G_2 = G$.
We expect that, at leading order, $G_n \equiv G$ in random tensor networks and in fixed geometry states.
However, for $n>2$ we find that other semiclassical states, including the AdS vacuum, exhibit spontaneous symmetry breaking of a $\ZZ_n$ replica symmetry, such that $G_n$ \emph{cannot} be interpreted in terms of a tripartition area in a backreacted geometry.
We consider the possibility that the family $G_n(A : B : C)$ has a natural analytic continuation (analogous to the R\'{e}nyi entropies~$S_n$ for noninteger $n$) which would potentially include a quantity~$G_1$ (analogous to the entanglement entropy~$S$) determined by the limit in which the number of replicas goes to one.
However we find no obvious reason to expect that this is true.

\subsection*{Related work}
During the preparation of this manuscript, the related work \cite{gadde2022multi} appeared.
In this work, the authors assume the family $G_n$ is indeed analytic, and proceed to study the hypothetical quantity $G_1$, which they call the multi-entropy.%
\footnote{They also consider generalizations to four or more regions.}
As discussed above, we are unclear as to whether the assumption of analyticity is true.
Moreover, the authors of \cite{gadde2022multi} assume (as in the seminal derivation of the Ryu-Takayanagi formula in~\cite{lewkowycz2013generalized}) that the dominant bulk saddle preserves the boundary replica symmetry.
We show explicitly in \cref{sec:generalize} that this assumption is false in the present context:
the dominant saddles for $G_3$ for three contiguous intervals in vacuum AdS$_3$ spontaneously break the boundary replica symmetry.
The area of a minimal tripartition was also discussed as an interesting bulk quantity in~\cite{Bao:2018gck, Harper:2019lff, harper2021hyperthreads}, although no explicit boundary dual was suggested.

We also note that our setup is closely related to~\cite{akers2022reflected,akers2022reflected2}, which study the reflected entropy for random tensor networks.
The replica trick considered in our work is a special case of the one used to compute the reflected entropy in these works.
Our results in \cref{sec:rtn} are new since we analyze an arbitrary graph (with a unique minimal tripartition), while~\cite{akers2022reflected,akers2022reflected2} only prove results for tensor networks with at most two bulk vertices (although \cite{akers2022reflected} includes conjectures about more general tensor networks). Results for arbitrary graphs have been independently obtained in forthcoming work by the same authors \cite{akersreflected3}.

\section{A replica trick for multipartite entanglement}\label{sec:multipartite entanglement replica trick}
The replica trick is a method for analyzing the entanglement structure of a quantum state~$\ket{\Psi}$ by computing the expectation values of permutation operators acting on subsystems of the state.
So for example, given a bipartite state $\ket{\Psi}_{AB}$ with reduced density matrix $\rho_A$, the R\'{e}nyi entropy is
\begin{align*}
    S_n(A)_{\ket{\Psi}}
= \frac{1}{1-n} \log \tr\mleft[\rho_A^n\mright] = \frac{1}{1-n} \log \bra{\Psi}^{\otimes n} \tau_A \ket{\Psi}^{\otimes n},
\end{align*}
where $\tau_A$ cyclically permutes the $n$ copies of subsystem $A$.
Similar replica tricks can be used to compute properties of a tripartite state $\ket{\Psi}_{ABC}$.
Two well-known examples are the entanglement negativity and the reflected entropy, which have been studied using the replica trick in holographic theories and for random tensor network states~\cite{aubrun2012partial,nezami2020multipartite,dong2021holographic,dong2022replica,dutta2021canonical,akers2022reflected,kudler2022negativity}.
The realignment criterion is another example of an entanglement criterion that can be studied using a replica trick, as was done in \cite{aubrun2012realigning} for random states.

\subsection{Definition of \texorpdfstring{$G(A:B:C)$}{G(A:B:C)} and basic properties}\label{subsec:basic G}
In this work we focus on a quantity defined using particular simple replica trick.
Consider the nontrivial elements of the Klein four-group:
\begin{align}\label{eq:permutations for G}
\begin{split}
    \pi^{(1)} &= (1 2)(3 4), \\
    \pi^{(2)} &= (1 3)(2 4), \\
    \pi^{(3)} &= (1 4)(2 3).
\end{split}
\end{align}
Then, for a tripartite pure quantum state $\ket{\Psi}_{ABC}$, we define
\begin{align}\label{eq:def g}
    G(A:B:C)_{\ket{\Psi}} &= -\frac{1}{2} \log Z(A:B:C)_{\ket{\Psi}},
\end{align}
where
\begin{align*}
    Z(A:B:C)_{\ket{\Psi}} &= \bra{\Psi}^{\otimes 4} \parens*{ \pi^{(1)}_A \ot \pi^{(2)}_B \ot \pi^{(3)}_C } \ket{\Psi}^{\otimes 4}.
\end{align*}
As mentioned earlier, in general, given a subsystem $A$ and a permutation $\pi \in S_n$ we denote by $\pi_A$ the operator which applies the permutation $\pi$ to the $n$ copies of the system $A$, while acting as the identity otherwise.
We write $G(A:B:C)$, leaving out the subscript~$\ket\Psi$, whenever the state is clear or unimportant.

We have the following basic properties, which follow immediately from the definition.

\begin{lem}\label{lem:properties G}
The quantity $G(A : B : C)$ defined in \cref{eq:def g} satisfies the following properties:
\begin{enumerate}
\item\label{it:permutation invariance}
It is invariant under permuting (or relabeling) the subsystems $A$, $B$, and $C$.
E.g., $G(A:B:C) = G(B:A:C)$ etc.

\item It is additive under tensor products:
if $\ket{\Psi}_{A_1 B_1 C_1}$ and $\ket{\Phi}_{A_2 B_2 C_2} $ are pure quantum states, then
\begin{align}\label{eq:tensor product}
  G(A_1 A_2:B_1 B_2:C_1 C_2)_{\ket{\Psi} \ot \ket{\Phi}} = G(A_1:B_1:C_1)_{\ket{\Psi}} + G(A_2:B_2:C_2)_{\ket{\Phi}}.
\end{align}

\item It is invariant under local isometries.
That is, if $U_A \colon \HH_A \to \HH_{A'}$, $U_B \colon \HH_B \to \HH_{B'}$, and $U_C \colon \HH_C \to \HH_{C'}$ are isometries, then
\begin{align*}
    G(A' : B' : C')_{(U_A \ot U_B \ot U_C) \ket{\Psi}} = G(A : B : C)_{\ket{\Psi}}.
\end{align*}
\end{enumerate}
\end{lem}

It can be convenient to write out $G(A : B : C)$ in a less symmetric way.
Since $\ket\Psi^{\ot n}$ is symmetric, for any permutations $\pi, \pi^{(1)}, \pi^{(2)}, \pi^{(3)} \in S_n$, we have
\begin{align*}
  \bra\Psi^{\ot n} \parens*{ \pi^{(1)}_A \ot \pi^{(2)}_B \ot \pi^{(3)}_C } \ket\Psi^{\ot n}
= \bra\Psi^{\ot n} \parens*{ (\pi \pi^{(1)})_A \ot (\pi \pi^{(2)})_B \ot (\pi \pi^{(3)})_C } \ket\Psi^{\ot n}.
\end{align*}
For $n=4$, and choosing the permutations $\pi^{(i)}$ to be those in~\cref{eq:permutations for G} and $\pi = \pi^{(3)}$, we obtain
\begin{align*}
   Z(A:B:C) &= \bra\Psi^{\ot 4} \parens*{ \pi^{(1)}_A \ot \pi^{(2)}_B \ot \pi^{(3)}_C } \ket\Psi^{\ot 4} \\
&= \bra\Psi^{\ot 4} \parens*{ (\pi^{(3)} \pi^{(1)})_A \ot (\pi^{(3)} \pi^{(2)})_B \ot \id_C } \ket\Psi^{\ot 4} \\
&= \bra\Psi^{\ot 4} \parens*{ \pi^{(2)}_A \ot \pi^{(1)}_B \ot \id_C } \ket\Psi^{\ot 4} \\
&= \tr \rho_{AB}^{\ot 4} \parens*{ \pi^{(2)}_A \ot \pi^{(1)}_B },
\end{align*}
where we used that $\pi^{(3)} \pi^{(1)} = \pi^{(2)}$ etc.
One can similarly write this quantity as the expectation value of $\pi^{(i)}_A \ot \pi^{(j)}_B$ for any $i\neq j$, and also for any other pair of subsystems, e.g.,
\begin{equation}\label{eq:other permutations}
  Z(A:B:C) = \bra\Psi^{\ot 4} \parens*{ \pi^{(1)}_A \ot \pi^{(2)}_B \ot \id_C } \ket\Psi^{\ot 4}
    = \tr \rho_{AB}^{\ot 4} \parens*{ \pi^{(1)}_A \ot \pi^{(2)}_B }.
\end{equation}

The above can be interpreted as follows:
think of $\rho_{AB}$ as a \emph{projected entangled pair state (PEPS)} tensor, with horizontal bond dimension~$d_A$, vertical bond dimension~$d_B$, and no physical degrees of freedom.
Then $Z(A:B:C)$ is obtained by placing four copies of $\rho_{AB}$ on a $2 \times 2$ lattice and contracting with periodic boundary conditions, as in \cref{fig:gpeps}~(a).
This in turn gives rise to another useful formula.
Denote by $X_2$ the matrix product operator on~$\HH_A^{\ot 2}$ defined by contracting two copies of $\rho_{AB}$ along the $B$-direction, i.e., $X_2 := \tr_{B^{\ot 2}}[\rho_{AB}^{\ot 2} (1\,2)_B]$.
Then, $X_2$ is Hermitian and we have
\begin{align}\label{eq:Z via MPO}
  Z(A:B:C)
= \tr \rho_{AB}^{\ot 4} \parens*{ \pi^{(1)}_A \ot \pi^{(2)}_B }
= \tr\mleft[ X_2^2 \mright]
= \tr\mleft[ X_2^\dagger X_2 \mright]
= \norm{X_2}_2^2,
\end{align}
which shows that $Z(A:B:C) > 0$.
This immediately implies that $G(A : B : C)$ is real.

\begin{figure}
\centering
\begin{subfigure}{.4\textwidth}
\begin{overpic}[width=.8\linewidth,grid=false]{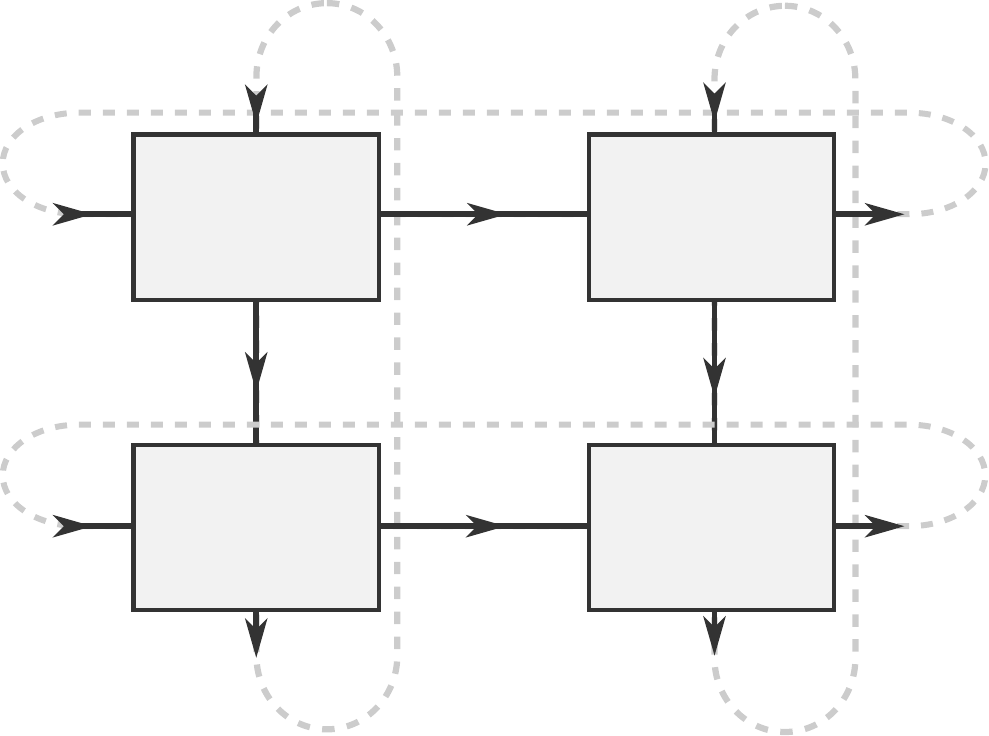}
  \put(20,50){$\rho_{AB}$} \put(65,50){$\rho_{AB}$}
  \put(20,19){$\rho_{AB}$} \put(65,19){$\rho_{AB}$}
  \put(46,55){\footnotesize{$B$}} \put(46,24){\footnotesize{$B$}}
  \put(29,35){\footnotesize{$A$}} \put(75,35){\footnotesize{$A$}}
  \put(-10,-10){(a)}
\end{overpic}
\vspace{0.55cm}
\end{subfigure}%
\hspace*{0.5cm}
\begin{subfigure}{.45\textwidth}
\begin{overpic}[width=.8\linewidth,grid=false]{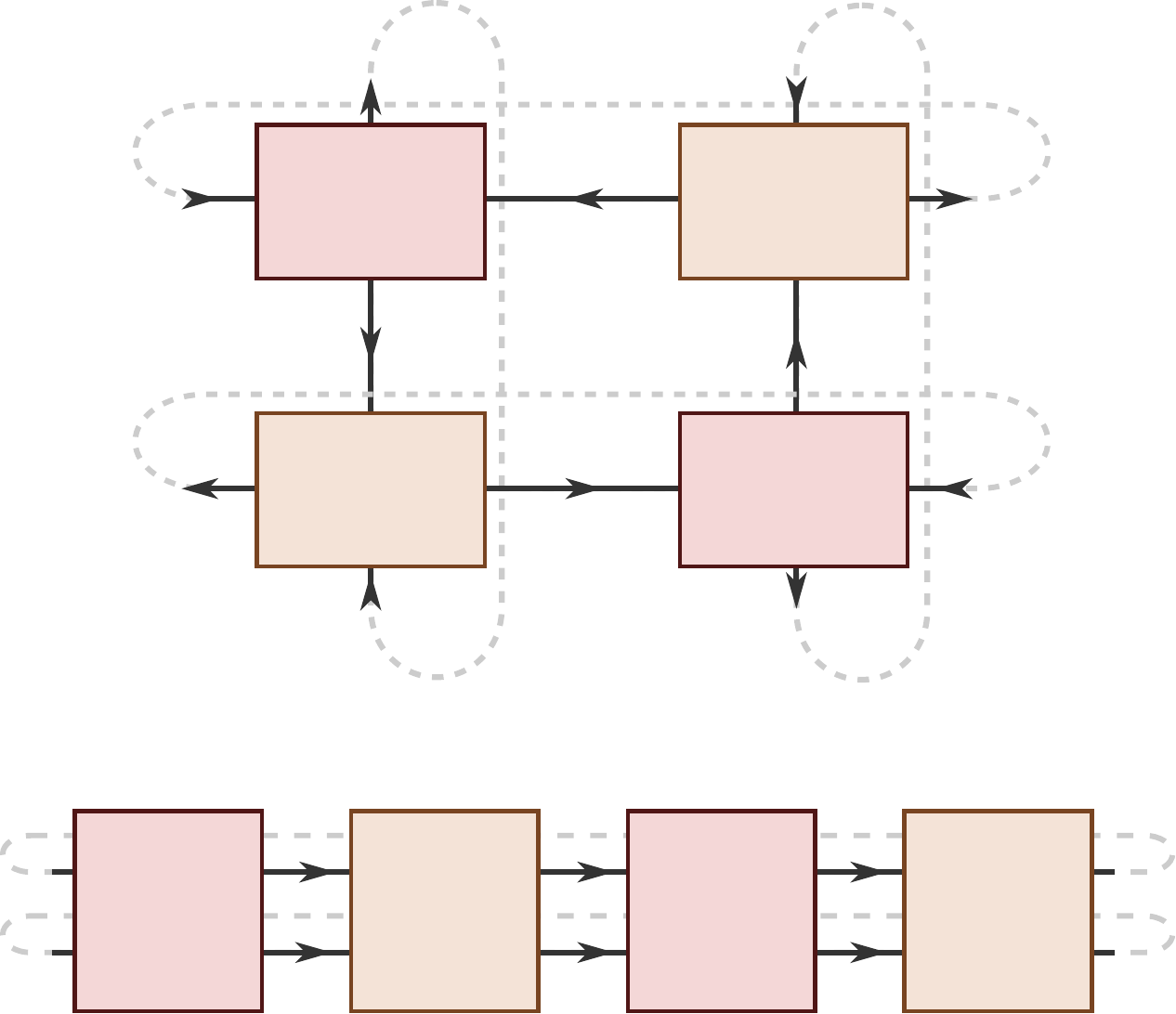}
  \put(25,67){$\rho_{AB}^R$} \put(62,67){$\rho_{AB}^{R,\dagger}$}
  \put(25,42){$\rho_{AB}^{R,\dagger}$} \put(62,42){$\rho_{AB}^R$}
  \put(48,72){\footnotesize{$B$}} \put(48,47){\footnotesize{$B$}}
  \put(34,55){\footnotesize{$A$}} \put(70,55){\footnotesize{$A$}}
  \put(50,27){\rotatebox{90}{$=$}}
  \put(-15,2){(b)}
  \put(9,6){$\rho_{AB}^R$} \put(56,6){$\rho_{AB}^R$}
  \put(33,6){$\rho_{AB}^{R,\dagger}$} \put(80,6){$\rho_{AB}^{R,\dagger}$}
  \put(47,-1){\footnotesize{$B$}} \put(94,-1){\footnotesize{$B$}}
  \put(24,-1){\footnotesize{$A$}} \put(71,-1){\footnotesize{$A$}}
\end{overpic}
\end{subfigure}
\caption{(a)~The quantity~$Z(A:B:C)$, defined in \cref{eq:other permutations} can be computed as a tensor network.
(b)~Alternatively, $Z(A:B:C)$ can be computed using the realignment~$\rho_{AB}^R$ and its adjoint~$\rho_{AB}^{R,\dagger}$, as in \cref{eq:Z via realignment}.}
\label{fig:gpeps}
\end{figure}
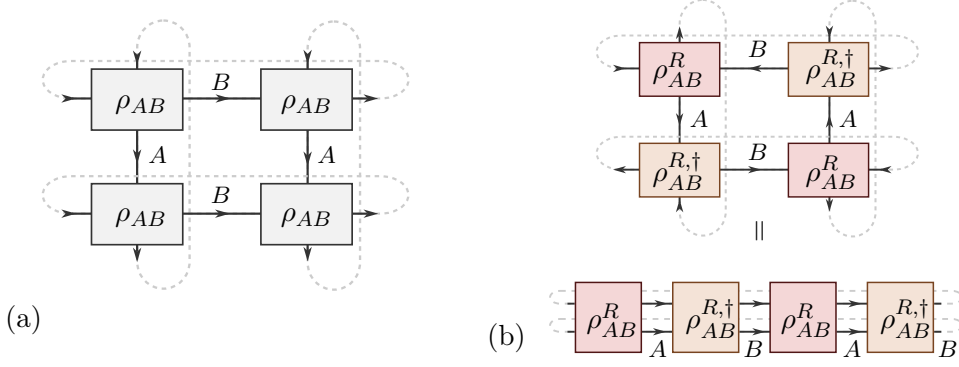

We can also express $G(A:B:C)$ in terms of the so-called \emph{realignment}~$\rho_{AB}^R$ of the reduced density matrix $\rho_{AB} = \tr_C\mleft[\proj{\Psi}\mright]$.
This is defined as follows in a basis independent way for an arbitrary operator $Y_{AB}$ on $\HH_A \ot \HH_B$:
we naturally identify $Y_{AB}$ with a vector in $\HH_A \ot \HH_B \ot \HH_A^* \ot \HH_B^*$.
By swapping the second and third tensor factor, we obtain a vector in $\HH_A \ot \HH_A^* \ot \HH_B \ot \HH_B^*$, which, interpreted as an operator $\HH_B^* \ot \HH_B \to \HH_A \ot \HH_A^*$, defines the realignment $Y_{AB}^R$.%
\footnote{If $Y_{AB}$ is expanded in the computational basis as $Y_{AB} = \sum_{i,k = 1}^{d_A} \sum_{j,l=1}^{d_B} \rho_{ij,kl} \ketbra{ij}{kl}$ and we identify $\HH_A^* \cong \HH_A$ and $\HH_B^* \cong \HH_B$ by the local computational bases, then $Y_{AB}^R = \sum_{i,k = 1}^{d_A} \sum_{j,l=1}^{d_B} \rho_{ij,kl} \ketbra{ik}{jl}$.}
Note that in general~$\rho_{AB}^R$ is not square, and even when it is, it need not be positive semidefinite (or even Hermitian).
The realignment can be used to detect entanglement in mixed states; if $\norm{\rho_{AB}^R}_1 > 1$ then~$\rho_{AB}$ must be entangled~\cite{chen2002matrix,rudolph2005further}.
It is easy to verify using \cref{eq:other permutations} and the fact that $\rho_{AB}$ is Hermitian, that, as shown in \cref{fig:gpeps}~(b),
\begin{align}\label{eq:Z via realignment}
  Z(A:B:C)
= \tr \rho_{AB}^{\ot 4} \parens*{ \pi^{(1)}_A \ot \pi^{(2)}_B }
= \tr\mleft[ \parens*{ ((\rho_{AB}^\dagger)^R)^\dagger \rho_{AB}^R }^2 \mright]
= \tr\mleft[ \parens*{ (\rho_{AB}^R)^\dagger \rho_{AB}^R }^2 \mright],
\end{align}
which can also be used to deduce that $Z(A:B:C)$ is positive and $G(A:B:C)$ is real.

We now proceed to prove basic bounds on $G(A:B:C)$.
We denote by $d_A$, $d_B$, and $d_C$ the dimensions of the Hilbert spaces of systems $A$, $B$, and $C$.

\begin{lem}\label{lem:bounds G}
For any tripartite pure state~$\ket{\Psi}_{ABC}$, we have
\begin{align*}
    0 \leq G(A : B : C)_{\ket{\Psi}} \leq \min\,\braces[\big]{ \log d_A + \log d_B, \log d_A + \log d_C, \log d_B + \log d_C }.
\end{align*}
\end{lem}
\begin{proof}
We already know that $Z(A:B:C)$ is real.
Since $\ket{\Psi}$ is a unit vector,
\begin{align*}
  Z(A:B:C)
= \bra\Psi^{\ot4} \parens*{ \pi^{(1)}_A \ot \pi^{(2)}_B \ot \pi^{(3)}_C } \ket\Psi^{\ot4}
\leq \norm*{\pi^{(1)}_A \ot \pi^{(2)}_B \ot \pi^{(3)}_C}_\infty
\leq 1,
\end{align*}
where $\norm{\cdot}_\infty$ denotes the operator norm.
Hence $G(A:B:C)_{\ket{\Psi}} \geq 0$.

To prove the upper bounds, we use \cref{eq:Z via MPO}, which states that $Z(A:B:C) = \norm{X_2}_2^2$ for $X_2 = \tr_{B^{\ot 2}}[\rho_{AB}^{\ot 2} (1\,2)_B]$, which is an operator on $\HH_A \ot \HH_A$.
By the Cauchy-Schwarz inequality,
\begin{align*}
  \frac1{d_B} \leq \tr \rho_B^2 = \tr X_2 \leq \norm{X_2} \sqrt{d_A^2} = \sqrt{Z(A:B:C)} \, d_A,
\end{align*}
which implies that $Z(A:B:C) \geq \frac1{d_A^2d_B^2}$ and hence $G(A:B:C) \leq \log (d_A d_B)$.
By permutation invariance as in \cref{lem:properties G}, this also implies the other two upper bounds.
\end{proof}

The proof of \cref{lem:bounds G} shows that in fact~$G(A:B:C) \leq S_0(A) + S_2(B)$ (and permutations).

Let us compute $G(A:B:C)$ for two basic examples.
The first example is the GHZ state of dimension $d$ on three parties $A$, $B$, and $C$
\begin{align}\label{eq:ghz}
    \ket{\GHZ_d} = \frac{1}{\sqrt d} \sum_{i=0}^{d-1} \ket{iii}.
\end{align}
Then it is easy to see using \cref{eq:other permutations} that the normalization contributes a factor $d^{-4}$ whereas the trace contributes a factor $d$ so
\begin{align*}
  Z(A:B:C)_{\ket{\GHZ_d}}
= \bra{\GHZ_d}^{\ot 4} \parens*{ \pi^{(1)}_A \ot \pi^{(2)}_B \ot \id_C } \ket{\GHZ_d}^{\ot 4}
= d^{-4} \cdot d
= d^{-3},
\end{align*}
and hence
\begin{align*}
    G(A:B:C)_{\ket{\GHZ_d}} = \frac{3}{2} \log d.
\end{align*}
A second example is the case of purely bipartite entanglement.
For $\ket{\Psi_{ABC}} = \ket{\psi}_{AB} \ot \ket0_C$,
\begin{align*}
  Z(A:B:C)_{\ket\Psi}
= \bra{\psi}^{\ot 4} \parens*{ \pi^{(1)}_A \ot \pi^{(2)}_B } \ket{\psi}^{\ot 4}
= \bra{\psi}^{\ot 4} \parens*{ (\pi^{(1)} \pi^{(2)})_A \ot \id_B } \ket{\psi}^{\ot 4}
= \parens*{ \tr \rho_A^2 }^2,
\end{align*}
and therefore
\begin{align*}
  G(A:B:C)
= S_2(A)
= \frac12 \parens*{ S_2(A) + S_2(B) + S_2(C) },
\end{align*}
where $S_2(\rho) = -\log\tr\rho^2$ is the R\'enyi-2 entropy.
It is easy to see using \cref{eq:tensor product} that the latter formula extends to an arbitrary purely bipartite entangled state of the form
\begin{align}\label{eq:bipartite}
  \ket\Psi_{ABC} = \ket{\psi_1}_{A_1 B_1} \ot \ket{\psi_2}_{A_2 C_1} \ot \ket{\psi_3}_{B_2 C_2}.
\end{align}
For any such state we have
\begin{align*}
  G(A:B:C)
= \frac12 \parens*{ S_2(A) + S_2(B) + S_2(C) }.
\end{align*}

A common tool in quantum information theory for one-shot entropic quantities is to perform \emph{smoothing}, which means that one optimizes a quantity of interest over all nearby states~\cite{tomamichel2015quantum}.
Recently, this has also found application in the context of holographic quantum gravity~\cite{akers2021leading,akers2022quantum} and random tensor networks~\cite{cheng2022random}.
For a pure state $\ket{\Psi}$ on a Hilbert space~$\HH$, let $B(\Psi, \eps)$ be the set of pure states $\ket{\widetilde\Psi}$ on $\HH$ which are close in trace distance, i.e., $T(\Psi,\widetilde\Psi) := \frac{1}{2}\norm{\ketbra{\widetilde\Psi}{\widetilde\Psi} - \ketbra{\Psi}{\Psi}}_1 \leq \eps$.
We introduce a smoothed version of $G(A:B:C)$ which depends on a parameter $0 < \eps < 1$ as
\begin{align}\label{eq:smoothed G}
    G^{\eps}(A : B : C)_{\ket{\Psi}} = \sup_{\ket{\widetilde\Psi} \in B(\Psi,\eps)} G(A : B : C)_{\ket{\widetilde\Psi}}.
\end{align}

\subsection{The Cayley distance on \texorpdfstring{$S_k$}{the symmetric group}}\label{subsec:cayley}
Given a permutation $\pi \in S_n$ we denote by $\abs{C(\pi)}$ the number of disjoint cycles that make up~$\pi$ (including fixed points).
We may define a distance function, the \emph{Cayley distance}, on the symmetric group $S_n$ in the following way:
\begin{align*}
    d(\pi,\sigma) = k - \abs{C(\pi^{-1}\sigma)}.
\end{align*}
One can show that equivalently $d(\pi,\sigma)$ is equal to the minimal number of transpositions one has to apply to transform $\pi$ into $\sigma$.
That is, $d$ is the Cayley distance with respect to the generating set of~$S_n$ that consists of all transpositions.

We say that a permutation $\pi \in S_k$ is on a \emph{geodesic} between $\pi_1, \pi_2 \in S_k$ if the triangle inequality is saturated
\begin{align*}
    d(\pi_1, \pi) + d(\pi, \pi_2) = d(\pi_1, \pi_2).
\end{align*}

The Cayley distance and geodesics on the symmetric group are closely related to the theory of non-crossing partitions~\cite{biane1997some,nica2006lectures} and is useful in random matrix theory~\cite{mingo2017free}.
In a similar fashion, in the replica trick for holographic states and random tensor network states, the Cayley distance plays a crucial role.
As we will discuss in more detail in \cref{sec:rtn} and \cref{sec:gravity}, replica trick computations involve labeling each bulk region or tensor by a permutation~$\pi$, whereas the boundary regions are labeled by fixed permutations that depend on the quantity of interest.
One then minimizes the ``action''
\begin{align*}
   S= \sum_{i, j} d(\pi_i,\pi_j) \, \abs{\gamma_{ij}}
\end{align*}
where $\abs{\gamma_{ij}}$ is the area (or logarithmic bond dimension in a tensor network) of the surface between bulk regions labeled by~$i$ and~$j$ respectively.

As an example, which has been worked out in detail in~\cite{dong2021holographic}, consider computing
\[ \bra\Psi^{\ot n} (\tau_A \ot \tau^{-1}_B \ot \id_C) \ket{\Psi}^{\ot n}, \]
where $\tau \in S_n$ is a cyclic permutation; an analytic continuation of this quantity can be used to compute the entanglement negativity.
Crucially, there exist permutations $\pi$ (corresponding to so-called ``non-crossing pairings'') which are at the same time on a geodesic between $\id$ and $\tau$, between $\id$ and $\tau^{-1}$, and between $\tau$ and $\tau^{-1}$.
In the configuration that minimizes the action, the entanglement wedge of $A$ (the region bounded by $A$ and the minimal-area surface homologous to $A$) is labeled by~$\tau$, the entanglement wedge of $B$ is labeled by~$\tau^{-1}$, the entanglement wedge of $C$ by~$\id$, and everything else by an arbitrary permutation $\pi$ corresponding to a non-crossing pairing.
As a result, the action depends in leading order only on the areas of minimal surfaces.

On the other hand, for the quantity $G(A : B : C)$, the permutations in~\cref{eq:permutations for G} are \emph{incompatible}:\label{subsec:incompat}
any permutation $\pi$ which is on a geodesic between a pair of $\pi^{(1)}$, $\pi^{(2)}$, and $\pi^{(3)}$ is not on a geodesic between any other pair.
Moreover, $d(\pi^{(1)}, \pi^{(2)}) = d(\pi^{(1)},\pi^{(3)}) = d(\pi^{(2)},\pi^{(3)}) = 2$.
This structure suggests that the optimal bulk glueing is given by a partitioning of the bulk into three regions $\Gamma_A$, $\Gamma_B$, and $\Gamma_C$ with respective boundaries $A$, $B$, and $C$, and labeling $\Gamma_A$ by $\pi^{(1)}$, $\Gamma_B$ by $\pi^{(2)}$, and $\Gamma_C$ by $\pi^{(3)}$.
The regions $\Gamma_A$, $\Gamma_B$ and $\Gamma_C$ should then be such that the total boundary between these three regions is minimal -- in other words, they should make up a minimal tripartition.
In \cref{sec:rtn} we show that this is indeed the case for a random tensor network.
In \cref{sec:gravity} we argue that the same is true for fixed area states in holographic gravity and that by using the smoothed quantity $G^\eps(A:B:C)$ this result extends to general semiclassical states.

\section{Random tensor network states}\label{sec:rtn}
Random tensor network states are a useful toy model of holographic quantum gravity~\cite{hayden2016holographic}.
They have been used to investigate entanglement properties of holographic states~\cite{hayden2016holographic,yang2016bidirectional,qi2017holographic,qi2018spacetime,penington2022replica,nezami2020multipartite,dong2021holographic,dong2022replica,akers2022reflected,dutta2021canonical,kudler2022negativity,qi2022holevo,apel2022holographic,cheng2022random,akers2022reflected2}.

They can be defined as follows, see \cite{hayden2016holographic,cheng2022random} for a more extensive discussion of the definition and basic properties.
We first choose an arbitrary graph with vertices $V$ and edges $E$, with~$V$ a disjoint union of two sets called the boundary vertices~$V_\partial$ and the bulk vertices~$V_b$.
We denote by $D$ the bond dimension, and for each edge $e = (vw)$ we have Hilbert spaces of half-edges $\HH_{e,v} = \HH_{e,w} = \CC^D$.
Let
\begin{align*}
    \ket{\phi_e} = \frac{1}{\sqrt{D}} \sum_{i=1}^D \ket{ii} \in \HH_{e,v} \ot \HH_{e,w}
\end{align*}
be a maximally entangled state along the edge $e$.
At each vertex we let $\HH_v$ be the Hilbert space of all half-edges at the vertex $v$, which has dimension $D^{d(v)}$ where $d(v)$ is the degree of the vertex $v$.
For each bulk vertex $v \in V_b$ we choose a uniformly Haar random tensor $\ket{\psi_v}$ in the associated Hilbert space, and we define the random tensor network state to be given by
\begin{align}\label{eq:rtn}
    \ket{\Psi} = \left(I_{V_\partial} \ot \bigotimes_{v \in V_b} D^{\frac12 d(v)}\bra{\psi_v} \right) \bigotimes_{e \in E} \ket{\phi_e},
\end{align}
which is normalized in expectation.
See \cref{fig:rtn}~(a) for an illustration.

It is well-known that one can use the replica trick to compute certain observables of the random tensor network state in expectation.
That is, given boundary regions $A_1, \dots, A_m$ whose disjoint union equals $V_\partial$, and permutations $\pi^{(1)}, \dots, \pi^{(m)} \in S_k$ one may compute
\begin{align*}
  \EE \bra{\Psi}^{\ot n} \parens*{ \pi^{(1)}_{A_1} \ot \dots \ot \pi^{(m)}_{A_m} } \ket{\Psi}^{\ot n}
\end{align*}
as the partition function of a classical spin model at inverse temperature $\log(D)$.
We describe this spin model.
The spins take values in the permutation group~$S_n$ for each vertex in~$V$.
Then the configurations of the spin model are assignments $\{\pi_v\}_{v \in V}$ where $\pi_v \in S_n$, and where we impose the boundary conditions that $\pi_v = \pi^{(i)}$ for $v \in A_i$.
The energy of a configuration is given by the sum of the Cayley distances (\cref{subsec:cayley}) along all edges in the graph
\begin{align*}
    E(\{\pi_v\}_{v \in V}) =  \sum_{e =(vw) \in E} d(\pi_v, \pi_w).
\end{align*}
Then one can show that
\begin{align}\label{eq:partition function}
  \EE \bra{\Psi}^{\ot n} \parens*{ \pi^{(1)}_{A_1} \ot \dots \ot \pi^{(m)}_{A_m} } \ket{\Psi}^{\ot n}
= \sum_{\{\pi_v\}_{v \in V}} e^{-\log(D)E(\{\pi_v\}_{v \in V})}
\end{align}
For large $D$ we can approximate \cref{eq:partition function} by considering only the configuration(s) that minimize the energy.
That is, if $E_0$ is the ground state energy and $N_0$ is the ground state degeneracy,
\begin{align}\label{eq:ground state approximation}
  \EE \bra{\Psi}^{\ot n} \parens*{ \pi^{(1)}_{A_1} \ot \dots \ot \pi^{(m)}_{A_m} } \ket{\Psi}^{\ot n}
= D^{-E_0} \parens*{ N_0 + \bigO(D^{-1}) }
\end{align}
since the energy function takes integer values.

The most well-known application is the case where there are two regions~$A_1 = A$ and~$A_2 = B = \bar A$, and we take $\pi^{(1)}$ to be the full cycle~$\tau$ of length~$n$ and $\pi^{(2)} = \id$.
This can be used to compute the R\'enyi entropy $S_n(\rho_A)$.
Indeed, one can show from \cref{eq:ground state approximation} that one has at leading order in the bond dimension~$D$~\cite{hayden2016holographic}
\begin{align*}
    S_n(\rho_A) \approx \abs{\gamma_A} \log D
\end{align*}
where $\gamma_A$ is the size of a minimal cut that separates $A$ from $\bar A$ in the graph.
This mimics the Ryu-Takayanagi formula of holographic gravity.

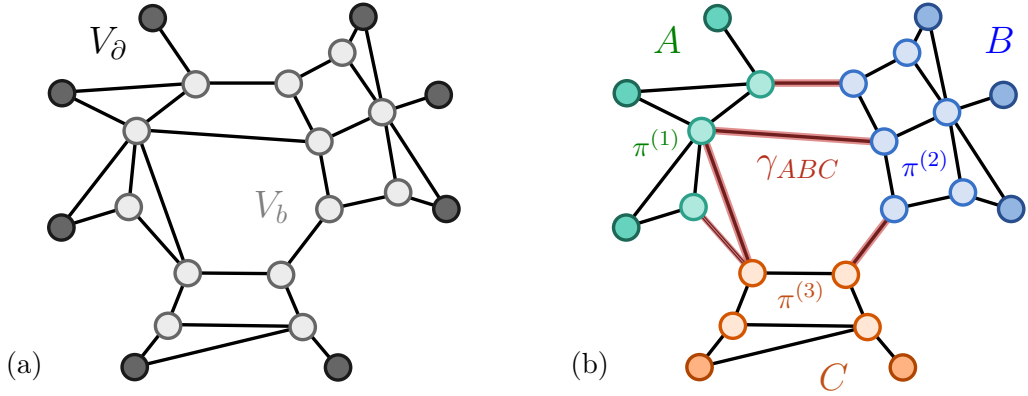
\begin{figure}
\centering
\begin{subfigure}{.45\textwidth}
\begin{overpic}[width=.8\linewidth,grid=false]{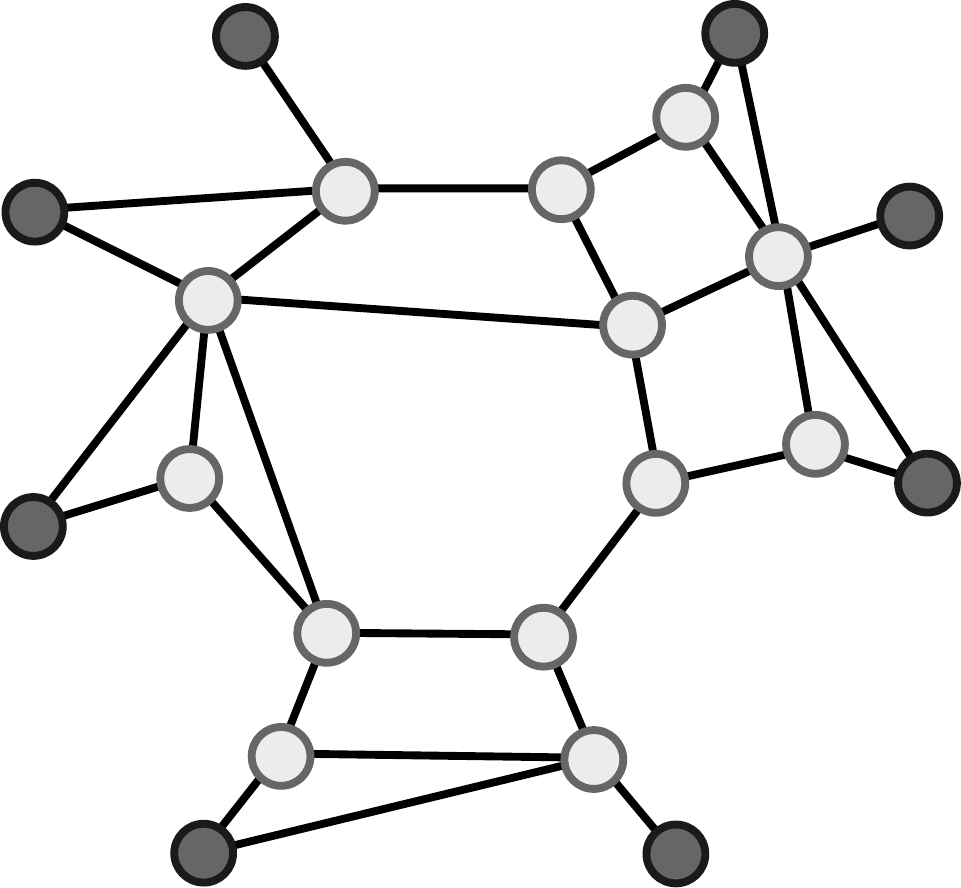}
    \put(10,80){\color{Black}{\Large{$V_{\partial}$}}}
    \put(50,40){\color{Gray}{\Large{$V_b$}}}
    \put(-10,2){(a)}
\end{overpic}
\end{subfigure}%
\hspace*{0.5cm}
\begin{subfigure}{.45\textwidth}
\begin{overpic}[width=.8\linewidth,grid=false]{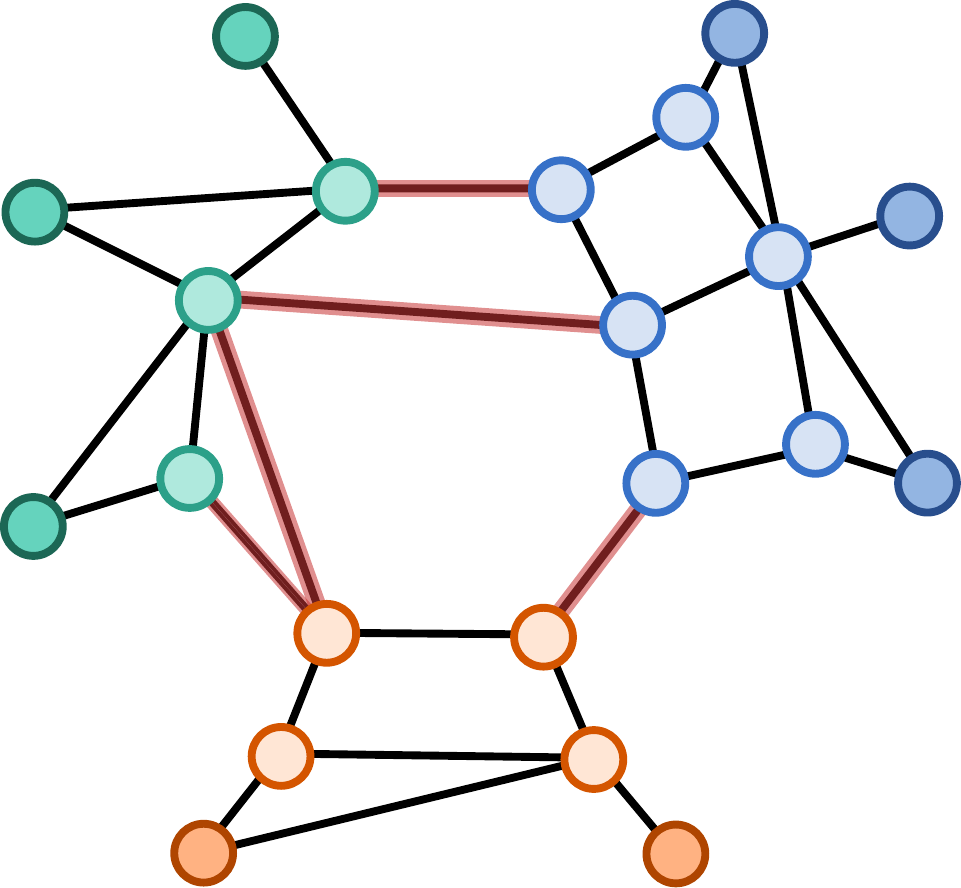}
    \put(10,80){\color{Green}{\Large{$A$}}}
    \put(90,80){\color{Blue}{\Large{$B$}}}
    \put(50,-2){\color{Bittersweet}{\Large{$C$}}}
    \put(5,55){\color{Green}{$\pi^{(1)}$}}
    \put(70,50){\color{Blue}{$\pi^{(2)}$}}
    \put(40,18){\color{Bittersweet}{$\pi^{(3)}$}}
    \put(35,51){\color{BrickRed}{\Large{$\gamma_{ABC}$}}}
    \put(-10,2){(b)}
\end{overpic}
\end{subfigure}

\caption{(a) A random tensor network state is constructed on a graph with bulk vertices~$V_b$ and boundary vertices~$V_{\partial}$.
Edges represent maximally entangled states, and we place random tensors at the bulk vertices.
The resulting state lives in a Hilbert space labeled by the boundary vertices.
(b)~Computing $G(A : B :C)$ for boundary subsystems~$A$, $B$, $C$ amounts to solving a spin model on the graph.
The optimal configuration is given by assigning the permutations~\cref{eq:permutations for G} to the respective parts of a minimal tripartition, as illustrated.
The corresponding energy is given by the size of the edge set~$\gamma_{ABC}$.\label{fig:rtn}}
\end{figure}

\subsection{The optimal configuration for \texorpdfstring{$G(A : B : C)$}{G} is a minimal tripartition}\label{subsec:opti}
We will now compute $G(A : B : C)$ for a random tensor network and show that it is computed by a minimal tripartition.
Given a partitioning of the boundary into three disjoint sets, $V_\partial = A \cup B \cup C$, a \emph{tripartition} (or \emph{3-terminal cut} or \emph{3-way cut}) for $A$, $B$, and $C$ is a partitioning of the vertices of the graph $V = \Gamma_A \cup \Gamma_B \cup \Gamma_C$ into disjoint subsets of vertices, such that $V_{\partial} \cap \Gamma_A = A$, $V_{\partial} \cap \Gamma_B = B$ and $V_{\partial} \cap \Gamma_C = C$.
Associated to each tripartition we have a set of edges~$\gamma_{AB}$ between $\Gamma_A$ and $\Gamma_B$ (that is, edges $(xy) \in E$ with $x \in \Gamma_A$ and $y \in \Gamma_B$) and similarly defined sets of edges~$\gamma_{BC}$ and~$\gamma_{AC}$.
Let $\gamma_{ABC} = \gamma_{AB} \cup \gamma_{BC} \cup \gamma_{AC}$.
A tripartition is \emph{minimal} if $\abs{\gamma_{ABC}} = \abs{\gamma_{AB}} + \abs{\gamma_{AB}} + \abs{\gamma_{AB}}$ is minimal among all tripartitions for $A$, $B$, and $C$.
In the following we assume that the minimal 3-terminal cut is unique.

To compute $G(A:B:C)$, we consider the replica trick in a situation where we have three boundary regions $A$, $B$, and $C$ and boundary conditions as in \cref{eq:permutations for G}.
We claim that the dominant configuration in the corresponding bulk spin model is determined by the minimal tripartition, as illustrated in \cref{fig:rtn}~(b).

\begin{thm}\label{prp:minimal tripartition}
Suppose that the graph has a minimal tripartition $V = \Gamma_A \cup \Gamma_B \cup \Gamma_C$ for $A$, $B$, and $C$.
Then for $\pi^{(1)} = (12)(34)$, $\pi^{(2)} = (13)(24)$, and $\pi^{(3)} = (14)(23)$ as in \cref{eq:permutations for G},
\begin{align} \label{eq:expectedgeoffitivity}
\EE Z(A:B:C)_{\ket\Psi}
= \EE \bra\Psi^{\ot 4} \parens*{ \pi^{(1)}_A \ot \pi^{(2)}_B \ot \pi^{(3)}_C } \ket\Psi^{\ot 4}
= D^{-2 \abs{\gamma_{ABC}}} \parens*{ 1 + \bigO\mleft(D^{-1}\mright) }.
\end{align}
\end{thm}

\begin{proof}
By \cref{eq:other permutations}, we may instead compute
\begin{align*}
  \EE \bra\Psi^{\ot 4} \parens*{ \pi^{(1)}_A \ot \pi^{(2)}_B \ot \id_C } \ket\Psi^{\ot 4}.
\end{align*}
We will show that the optimal configuration $\{\pi_v\}_{v \in V}$ for these boundary conditions is given by $\pi_v = \pi^{(1)}$ for $v \in \Gamma_A$, $\pi_v = \pi^{(2)}$ for $v \in \Gamma_B$ and $\pi_v = \id$ for $v \in \Gamma_C$.
Since $d(\pi^{(1)},\pi^{(2)}) = d(\pi^{(1)},\id) = d(\pi^{(2)},\id) = 2$ this configuration has energy
\begin{align*}
    E(\{\pi_v\}_{v \in V}) = 2\abs{\gamma_{ABC}},
\end{align*}
so this will imply \cref{eq:expectedgeoffitivity} at once.

To see that the configuration described above is the optimal one and that there are no other ones, consider an arbitrary configuration $\{\pi_v\}_{v \in V}$ subject to the boundary conditions.
For $\pi \in S_4$, denote the corresponding domain by
\begin{align*}
    D(\pi) = \{v \in V : \pi_v = \pi \}.
\end{align*}
Then we define the following two tripartitions:
\begin{align*}
    \Delta_A &= D(\pi^{(1)}), &
    \Delta_B &= D(\pi^{(2)}) \;\;\cup\;\; \bigcup_{\mathclap{\substack{\pi \,:\, d(\pi,\pi^{(2)}) \\ = d(\pi,\pi^{(1)}) = 1}}} D(\pi), &
    \Delta_C &= V \setminus (\Delta_A \cup \Delta_B)
\end{align*}
and
\begin{align*}
    \tilde\Delta_A &= V \setminus (\tilde\Delta_B \cup \tilde\Delta_C), &
    \tilde\Delta_B &= D(\pi^{(2)}) \;\;\cup\;\; \bigcup_{\mathclap{\substack{\pi \,:\, d(\pi,\pi^{(2)}) = 1, \\ d(\pi,\pi^{(1)}) \neq 1}}} D(\pi), &
    \tilde\Delta_C &= D(\id).
\end{align*}
Denote the associated edge sets by $\delta_{ABC}$ and $\tilde \delta_{ABC}$, respectively.
We claim that
\begin{align}\label{eq:lower bound energy}
  E(\{\pi_v\}_{v \in V})
\geq \sum_{e = (vw) \in \delta \cup \tilde\delta} d(\pi_v, \pi_w)
\geq \abs{\delta_{ABC}} + \abs{\tilde \delta_{ABC}}.
\end{align}
This would confirm the optimality of the configuration described above, as $\gamma_{ABC}$ is the edge set of a \emph{minimal} tripartition and hence we must have
\begin{align}\label{eq:lower bound energy 2}
    E(\{\pi_v\}_{v \in V}) \geq \abs{\delta_{ABC}} + \abs{\tilde\delta_{ABC}} \geq 2\abs{\gamma_{ABC}}.
\end{align}
In fact, the above also implies that the described configuration is the \emph{unique} optimal configuration.
To see this, note that since we assumed that the minimal tripartition is unique, we can have equality in \cref{eq:lower bound energy 2} only if $\Delta_A = \tilde \Delta_A = \Gamma_A$ etc., and hence~$\delta_{ABC} = \tilde \delta_{ABC} = \gamma_{ABC}$.
Moreover, in that case we can have equality in \cref{eq:lower bound energy} only if for all edges~$(vw)$ which are not in $\gamma_{ABC}$ we have $\pi_v = \pi_w$.
This implies that a configuration with equality in \cref{eq:lower bound energy 2} and hence in \cref{eq:lower bound energy} has to be constant on the connected components of $\Gamma_A$, $\Gamma_B$ and $\Gamma_C$.
While $\Gamma_A$ need not itself be connected, each connected component must have a nonempty intersection with $A$ (otherwise, it is easy to see that assigning this component to~$\Gamma_B$ or~$\Gamma_C$ gives rise to another minimal tripartition which contradicts the uniqueness of the minimal tripartition), and similarly for $\Gamma_B$ and $\Gamma_C$.
Given the boundary conditions on $A$, $B$, and~$C$ this implies that the unique configuration with equality in \cref{eq:lower bound energy 2} has to be constant on~$\Gamma_A$, $\Gamma_B$, and~$\Gamma_C$, with $\pi_v = \pi^{(1)}$ for~$v \in \Gamma_A$, $\pi_v = \pi^{(2)}$ for~$v \in \Gamma_B$, and $\pi_v = \id$ for~$v \in \Gamma_C$.

We now set out to prove \cref{eq:lower bound energy}.
First we note that
\begin{align}
  E(\{\pi_v\}_{v \in V})
\label{eq:split energy}
&= \sum_{e = (vw) \in \delta \setminus \tilde \delta} d(\pi_v, \pi_w) + \sum_{e = (vw) \in \tilde \delta \setminus \delta} d(\pi_v, \pi_w) + \sum_{e = (vw) \in \delta \cap \tilde \delta} d(\pi_v, \pi_w).
\end{align}
From \cref{eq:split energy} it follows that if for each edge $(vw) \in \delta \cap \tilde \delta$ we have $d(\pi_v,\pi_w) \geq 2$, then \cref{eq:lower bound energy} is valid.
So, we consider $(vw) \in \delta \cap \tilde \delta$.
There must be at least one of $v,w \in \Delta_A \cup \Delta_B$ and one of $v,w \in \tilde \Delta_B \cup \tilde \Delta_C$.
We distinguish the following cases:
\begin{enumerate}
\item \emph{One of $v,w \in \Delta_A$ and one of $v,w \in \tilde \Delta_C$}:
Note that $\Delta_A$ and $\tilde\Delta_C$ are disjoint, so we may assume without loss of generality that $v \in \Delta_A$ and $w \in \tilde \Delta_C$.
This means $\pi_v = \pi^{(1)}$ and $\pi_w = \id$, and hence $d(\pi_v,\pi_w) = 2$.

\item \emph{One of $v,w \in \Delta_A$ and one of $v,w \in \tilde \Delta_B$}:
Note that $\Delta_A$ and $\tilde\Delta_B$ are likewise disjoint (if $d(\pi_w,\pi^{(2)}) = 1$ then we cannot have $\pi_w = \pi^{(1)}$), so we may assume without loss of generality that $v \in \Delta_A$ and $w \in \tilde\Delta_B$.
Then $\pi_v = \pi^{(1)}$, and $\pi_w$ either equals $\pi^{(2)}$ or satisfies $d(\pi_w,\pi^{(1)}) \neq 1$ and is not the same as $\pi_v$.
Clearly, in both cases $d(\pi_v,\pi_w) \geq 2$.

\item \emph{One of $v,w \in \Delta_B$ and one of $v,w \in \tilde \Delta_C$}:
Again $\Delta_B$ and $\tilde\Delta_C$ are disjoint, so we may assume without loss of generality that $v \in \tilde \Delta_C$ and $w \in \Delta_B$.
Therefore $\pi_v = \id$.
On the other hand, $\pi_w$ either equals $\pi^{(2)}$ (in which case $d(\pi_v,\pi_w) = d(\id,\pi^{(2)}) = 2$) or satisfies $d(\pi_w,\pi^{(2)}) = d(\pi_w,\pi^{(1)}) = 1$ (but then $d(\pi_v,\pi_w) \geq 2$ since no permutation is simultaneously on geodesics between $\id$, $\pi^{(1)}$ and $\pi_2$, as discussed earlier).

\item \emph{One of $v,w \in \Delta_B$ and one of $v,w \in \tilde \Delta_B$:}
Note that now $\Delta_B \cap \tilde\Delta_B = D(\pi^{(2)})$.
There are two cases to consider.
First suppose that one of the two permutations equals $\pi^{(2)}$, and let us assume without loss of generality this is~$\pi_v$.
Then, since $(vw) \in \delta \cap \tilde\delta$, we must have $\pi_w \not\in \Delta_B \cup \tilde\Delta_B$, and hence $d(\pi_w, \pi^{(2)}) \geq 2$.
The other possibility is that neither permutation equals $\pi^{(2)}$.
In that case, one of the vertices must be in $\Delta_B$ and the other in $\tilde\Delta_B$.
Without loss of generality, we may assume that $v \in \Delta_B$ and $w \in \tilde\Delta_B$.
Thus, $d(\pi_v,\pi^{(2)}) = d(\pi_v,\pi^{(1)}) = 1$ and $d(\pi_w,\pi^{(2)}) = 1$, $d(\pi_w,\pi^{(1)}) \neq 1$.
Since $\pi_v$ and~$\pi_w$ have the same parity, we cannot have $d(\pi_v,\pi_w) = 1$.
Hence $d(\pi_v, \pi_w) = 2$.
\qedhere
\end{enumerate}
\end{proof}

\Cref{prp:minimal tripartition} computes an expectation value, and a natural question is whether the random variable $\bra{\Psi}^{\ot 4}(\pi^{(1)}_A \ot \pi^{(2)}_B \ot \pi^{(3)}_C) \ket{\Psi}^{\ot 4}$ concentrates around its expected value.
To see that this is indeed the case, one can compute the variance, similar as to one does for the entropy in random tensor networks~\cite{hayden2016holographic,nezami2020multipartite,cheng2022random}.
To do so, we note that
\begin{align*}
    \left(\bra{\Psi}^{\ot 4}(\pi^{(1)}_A \ot \pi^{(2)}_B \ot \pi^{(3)}_C) \ket{\Psi}^{\ot 4}\right)^2 = \bra{\Psi}^{\ot 8}(\tilde \pi^{(1)}_A \ot \tilde\pi^{(2)}_B \ot \tilde \pi^{(3)}_C) \ket{\Psi}^{\ot 8}
\end{align*}
where $\tilde \pi^{(1)} = (1 2)(3 4)(5 6)(7 8)$, $\tilde \pi^{(2)} = (1 3)(2 4)(5 7)(6 8)$ and $\tilde \pi^{(3)} = (1 4)(2 3)(5 8)(6 7)$.
By similar reasoning as in the proof of \cref{prp:minimal tripartition} one finds that
\begin{align*}
    \EE \left(\bra{\Psi}^{\ot 4}(\pi^{(1)}_A \ot \pi^{(2)}_B \ot \pi^{(3)}_C) \ket{\Psi}^{\ot 4}\right)^2
  = D^{-4\abs{\gamma_{ABC}}} \parens*{ 1 + \bigO(D^{-1}) }.
\end{align*}
This implies that for large bond dimension~$D$, with high probability,
\begin{align*}
  G(A : B : C) \approx \abs{\gamma_{ABC}}\log D.
\end{align*}

\subsection{Random tensor networks with bulk entropy}\label{sec:rtn bulk}
It is also possible to model the entropy of bulk quantum fields in random tensor networks~\cite{hayden2016holographic}.
This can be done by adapting the construction in \cref{eq:rtn} as follows.
One allows an additional bulk degree of freedom at each bulk vertex, so $\HH_v = \HH_{v,\bulk} \ot \bigotimes_{e = (vw)} \HH_{e,v}$, considers a state $\ket{\psi_{\bulk}} \in \bigotimes_{v \in V} \HH_{v,\bulk}$, and defines the boundary random tensor network state as
\begin{align*}
    \ket{\Psi}
  = \left(I_{V_\partial} \ot \bigotimes_{v \in V_b} D(v)^{\frac12} \bra{\psi_v} \right) \left( \ket{\psi_{\bulk}} \ot \bigotimes_{e \in E} \ket{\phi_e} \right),
\end{align*}
where $D(v) = \dim \HH_{v}$, which is a straightforward generalization of \cref{eq:rtn}.
This also gives rise to a classical statistical mechanics model for $\EE \bra\Psi^{\ot 4}(\pi^{(1)}_A \ot \pi^{(2)}_B \ot \id_C) \ket\Psi^{\ot 4}$, see~\cite{hayden2016holographic,cheng2022random} for details.
In the regime where $\dim \HH_{v,\bulk} \ll D$ it is easy to see that the optimal configuration is independent of the bulk state and we obtain
\begin{align*}
&\quad \EE \bra\Psi^{\ot 4} \parens*{ \pi^{(1)}_A \ot \pi^{(2)}_B \ot \pi^{(3)}_C } \ket\Psi^{\ot 4} \\
&= D^{-2 \abs{\gamma_{ABC}}} \parens*{ \bra{\psi_\text{bulk}}^{\ot 4} \parens*{ \pi^{(1)}_{\Gamma_A} \ot \pi^{(2)}_{\Gamma_B} \ot \pi^{(3)}_{\Gamma_C} } \ket{\psi_\text{bulk}}^{\ot 4} + \bigO\mleft(D^{-1}\mright) },
\end{align*}
where $V = \Gamma_A \cup \Gamma_B \cup \Gamma_C$ is the minimal tripartition.
Hence, for large~$D$,
\begin{align}\label{eq:rtns with bulk}
    G(A : B : C)_{\ket{\Psi}} \approx \abs{\gamma_{ABC}}\log D + G(\Gamma_A : \Gamma_B : \Gamma_C)_{\ket{\psi_{\bulk}}},
\end{align}
which shows~$G$ is corrected by same quantity $G$ but computed for the bulk states and the minimal tripartition.
This is analogous to the FLM correction of holographic entropy~\cite{faulkner2013quantum}.


\section{Holographic gravity}\label{sec:gravity}
We have seen how to compute $G(A:B:C)$ for a random tensor network state.
We would like to use these insights to study the same replica trick for holographic gravity states.
We start with fixed-area states, which are closely related to random tensor network states, and therefore $G(A:B:C)$ is computed by a bulk minimal tripartition.
Next we proceed to general semiclassical states, where we see that the smoothed version $G^{\eps}(A:B:C)$ is computed by a bulk minimal tripartition.

\subsection{Fixed-geometry states}\label{sec:fixedarea}
There is a very close relationship between replica computations of the expectation value of a products of permutation operators on boundary subsystems in random tensor network states and the analogous computations for gravitational states in AdS/CFT where the bulk geometry or, at least, aspects of the bulk geometry are fixed~\cite{dong2019flat,akers2019holographic,penington2022replica}.

Since most of the argument is standard, we will only sketch it here, and refer readers to the original literature for more details.
We consider a boundary state $\ket{\Psi}$, prepared by a Euclidean path integral and postselected onto some fixed bulk spatial geometry $g$ on the time-reflection symmetric slice.
To carry out the desired replica computation on $k$ replices we glue the boundary Euclidean path integrals preparing $n$ bras $\bra{\Psi}$ and $n$ kets $\ket{\Psi}$ to each other -- with the desired permutations of boundary subregions inserted.
This produces a boundary partition function $Z_\mathrm{repl}$ that computes the expectation value of the permutation up to normalization, as shown in \cref{fig:replica trick}.

In general, a state $\ket{\Psi}$ prepared by a Euclidean path integral will not be normalized.
When computing $G(A:B:C)$ (or other replica trick quantities such as R\'{e}nyi entropies), expectation values should be computed using the normalized state $\ket{\widehat\Psi} = \ket{\Psi} / \sqrt{\braket{\Psi|\Psi}}$.
The normalization can also be computed using a partition function $Z_1 = \braket{\Psi|\Psi}$ formed by gluing together a single pair of bra and ket partition functions.
The normalized expectation value of a permutation operator on $n$ bras and $n$ kets is given by
\begin{align} \label{eq:normalpf}
    \widehat Z_\mathrm{repl} = \frac{Z_\mathrm{repl}}{Z_1^n}.
\end{align}

\begin{figure}
\centering
\begin{subfigure}{.4\textwidth}
\begin{overpic}[width=.8\linewidth,grid=false]{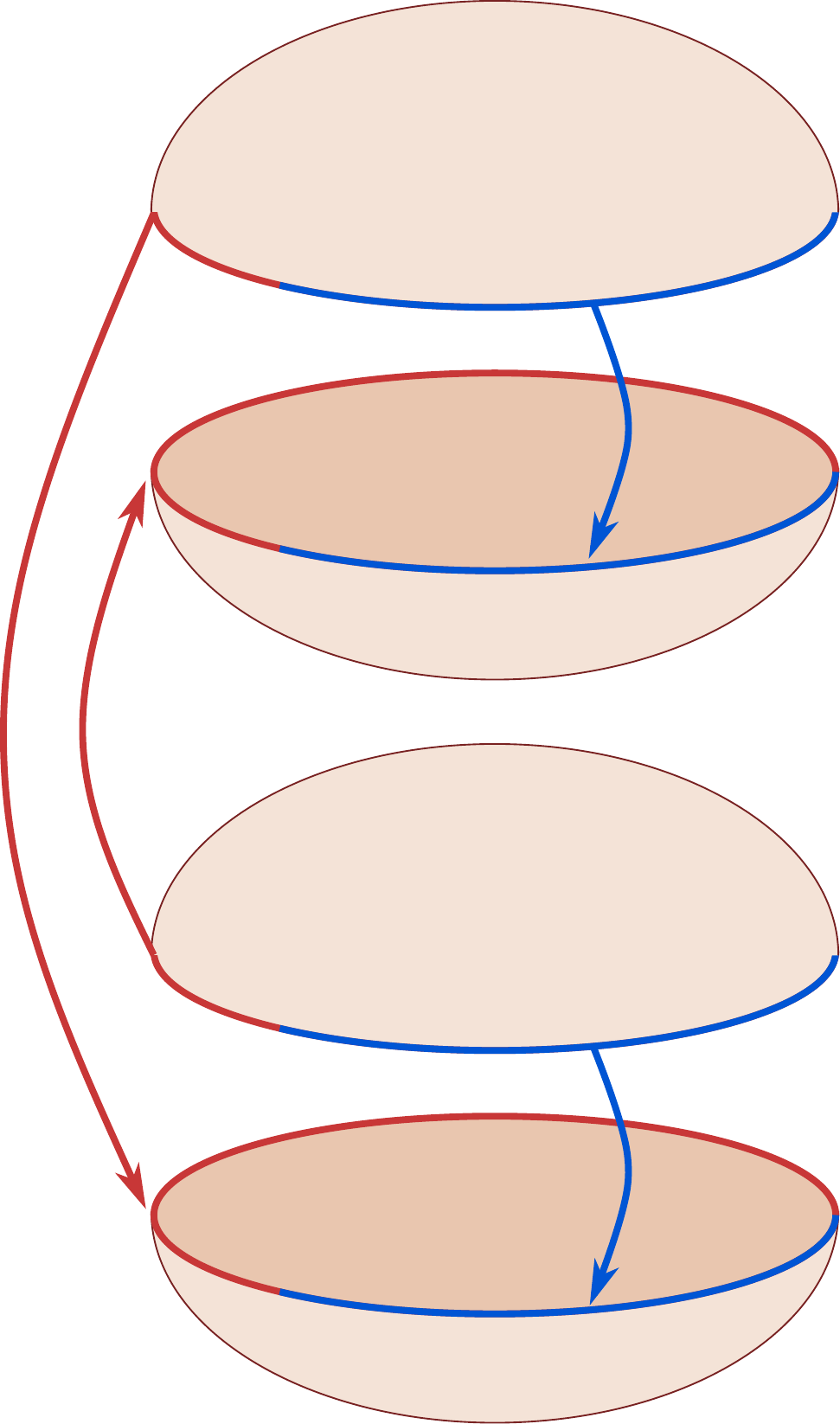}
    \put(10,73){\color{Bittersweet}{\Large{$A$}}}
    \put(30,62){\color{Blue}{\Large{$\bar{A}$}}}
    \put(31,34){$\bra{\Psi}$} \put(31,13){$\ket{\Psi}$}
    \put(-5,2){(a)}
\end{overpic}
\end{subfigure}%
\hspace*{0.5cm}
\begin{subfigure}{.5\textwidth}
\begin{overpic}[width=.8\linewidth,grid=false]{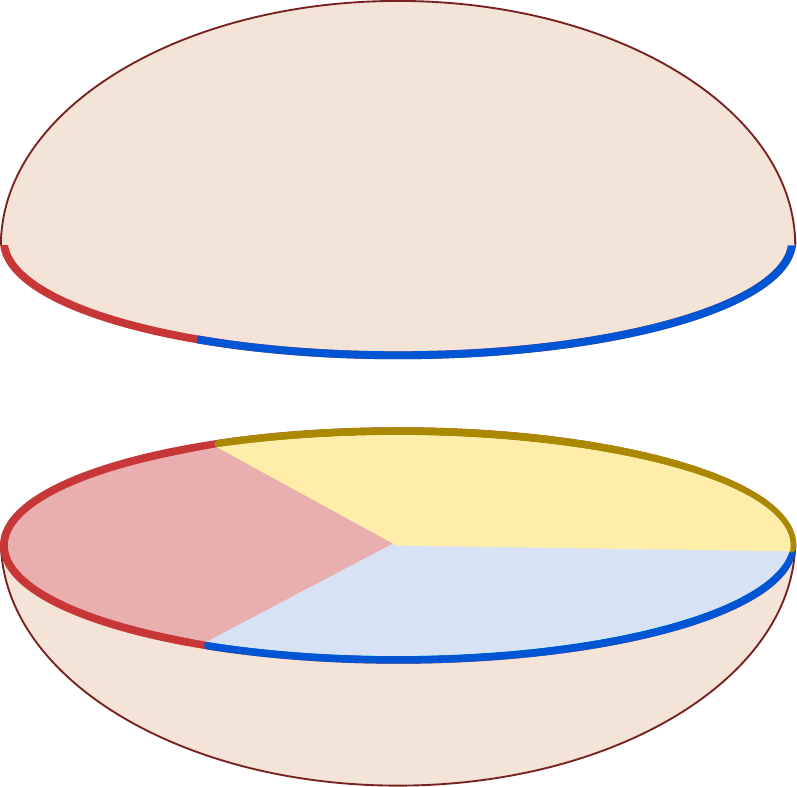}
    \put(-10,30){\color{Bittersweet}{\Large{$A$}}}
    \put(90,5){\color{Blue}{\Large{$B$}}}
    \put(80,45){\color{BurntOrange}{\Large{$C$}}}
    \put(18,28){\color{Bittersweet}{$\pi^{(1)}$}}
    \put(57,21){\color{Blue}{$\pi^{(2)}$}}
    \put(57,34){\color{BurntOrange}{$\pi^{(3)}$}}
    \put(94,90){\Large{$\times 4$}}
    \put(-10,2){(b)}
\end{overpic}
\end{subfigure}

\caption{(a)~Replica trick for the R\'enyi-2 entropy, which uses $n = 2$ copies and permutation~$(1 2)$ on region~$A$ and $\id$ on the complement $\bar{A}$.
(b)~The replica trick for $G(A : B : C)$ uses $n=4$ copies and~$\pi^{(1)} = (12)(34)$ on~$A$, $\pi^{(2)} = (13)(24)$ on~$B$, and $\pi^{(3)} = (14)(23)$ on~$C$.
The dominant bulk glueing prescription for a fixed geometry state is given by a minimal tripartition.}
\label{fig:replica trick}
\end{figure}
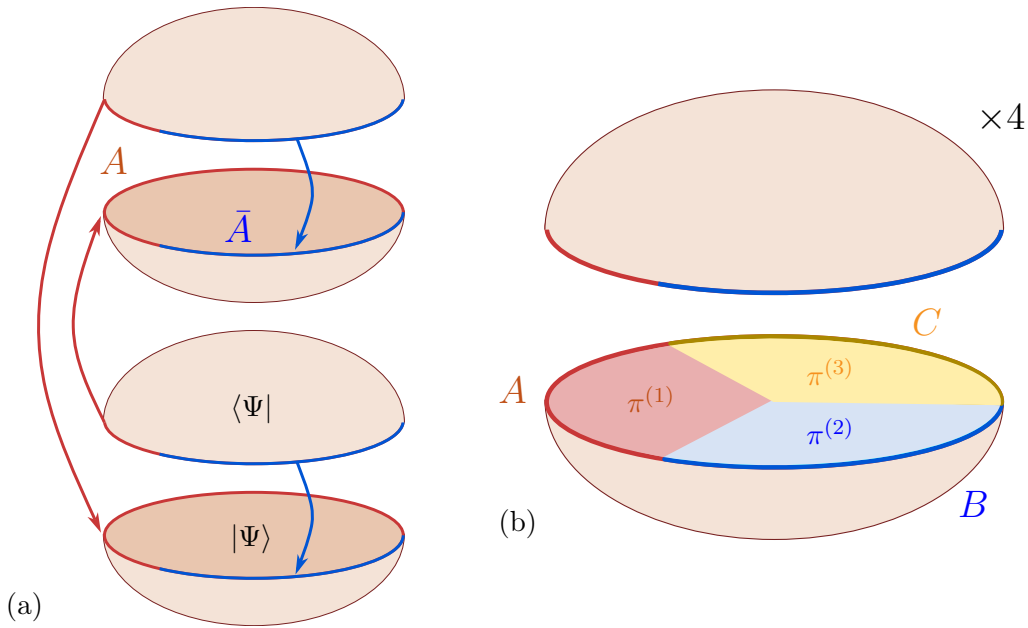

The AdS/CFT dictionary maps this replicated boundary partition function to a bulk partition function with asymptotic boundary conditions given by the boundary geometry (and boundary matter sources) obtained through this gluing procedure.
In the limit $G_N \to 0$, this bulk partition function can be evaluated semiclassically.
For general boundary states the saddle points of the bulk partition function are hard to find, but for states where the bulk geometry on a spatial slice is fixed, they turn out to have a simple form.

Consider as a toy example a single particle path integral where the position $x$ of the particle at some time $t_0$ is measured to be $x_0$.
Because of the measurement, we do not integrate over $x$ at time $t_0$, and hence a saddle point of the path integral does not need to obey the equation of motion for $x$ at time $t_0$.
Instead the conjugate momentum $p = m \partial_t x$ may jump discontinuously at time $t_0$ as needed in order to match the initial and final boundary conditions, along with the condition $x(t_0) = x_0$.

Exactly the same happens in the gravitational path integral.
Because the bulk geometry on a spatial slice in each bra or ket has been measured, the saddle point of the bulk partition function does not have to obey the equations of motion at the measured slice, and in particular the conjugate momenta to the spatial geometry, which involve time derivatives of the spatial metric, may be discontinuous across the slice.
As a result, the bulk geometry connecting the fixed geometry slice to each asymptotic boundary piece of the asymptotic boundary preparing a single bra or ket, does not care about the gluing procedure at all.
The same geometry is present in both the replica trick partition function $Z_\mathrm{repl}$ and in the partition function $Z_1$ used to compute the normalization $\braket{\Psi|\Psi}$.

The only question that remains in order to determine the entire bulk saddle, is how to glue together the fixed geometry slices of the $n$ bras and $n$ kets.
The choice of how to do this gluing in practice amounts to a choice of map from the bulk spatial geometry to the permutation group -- and is analogous to the choice over permutations associated to each vertex in the statistical mechanics model~\cref{eq:partition function} computing expectation values in a random tensor network.
In other words, such a gluing corresponds precisely to a division of the bulk geometry into domains~$D(\pi)$ for~$\pi \in S_n$ where we glue along~$\pi$ in the region~$D(\pi)$, with appropriate boundary conditions (i.e., if a boundary subregion~$A$ is assigned a permutation~$\pi$ in the CFT replica trick, the bulk region~$D(\pi)$ has region~$A$ as its conformal boundary).

To find the dominant saddle point we need to find the map that minimizes the total gravitational action.
This action has two contributions.
The first comes from the spacetime away from the fixed geometry slice.
This contribution is exactly $n$ times the action of the saddle point computing $Z_1 = \braket{\Psi|\Psi}$.
It therefore is (a)~independent of the choice of domains $D(\pi)$, and (b)~exactly cancels the action of the denominator $Z_1^n$ in the normalized formula~\eqref{eq:normalpf}.

The second contribution involves singularities in the geometry (from discontinuities in the extrinsic curvature) on the fixed geometry slices that differ in the numerator and denominator of \cref{eq:normalpf}.
These take the form of conical singularities at domain walls in the permutation map, i.e., the boundaries between regions $D(\pi)$ and $D(\sigma)$ for different permutations.%
\footnote{In fact, since these conical singularities are the only place where the equations of motion are not satisfied, it is sufficient to only fix the total area of each domain wall, rather than the entire bulk spatial geometry, in order to have the saddle point spacetime geometry not backreact relative to the $\braket{\Psi|\Psi}$ computation.
For this reason, it is common to talk about a correspondence between tensor networks and fixed-area (rather than fixed-geometry) states.
Of course, to determine the surfaces whose areas need to be fixed, one needs to know in advance where the domain walls will be!}
The action of these conical singularities is given by
\begin{align*}
S_{\text{con}} = \frac{1}{8 \pi G_N} \int dy^{d-2} \sqrt{h} \, \parens*{ \chi(y) - 2\pi }
\end{align*}
where the coordinates $y$ parameterise the domain wall, $dy^{d-2}\sqrt{h}$ is the volume form, and~$\chi(y)-2\pi$ is the conical excess at $y$.
For full generality we allow the original saddle point geometry for $Z_1$ to have a conical singularity with excess $\chi_0(y) - 2\pi$.
In the denominator of \cref{eq:normalpf}, we have $n$ copies of this conical singularity.
In the numerator however, the gluing rules mean that the number of copies of the singularity is given by the number of cycles $\abs{C(\pi^{-1}\sigma)} = n - d(\sigma,\pi)$ in the permutation $\pi^{-1}\sigma$ associated to the domain between $D(\pi)$ and $D(\sigma)$.
The conical singularity associated to a cycle $c$ has excess $(\ell_c\chi_0(y) - 2 \pi)$ where $\ell_c$ is the length of the cycle.
The additional action of the numerator relative to the denominator is therefore
\begin{align}\label{eq:DeltaScon}
  \Delta S_{\text{con}}
= \frac{1}{8\pi G_N}\sum_c \int dy^{d-2} \sqrt{h} \, 2 \pi \parens*{ \ell_c - 1}
= d(\pi,\sigma) \frac{\area(\pi,\sigma)}{4 G_N}.
\end{align}
where the sum is over all cycles in the permutation~$\pi^{-1}\sigma$ and $\area(\pi,\sigma)$ is the area of the domain wall between $D(\pi)$ and $D(\sigma)$; see~\cite{dong2019flat,akers2019holographic,penington2022replica} for details in this derivation.
This is the only part of the action that does not exactly cancel between the numerator and denominator of \cref{eq:normalpf}, and the only part that depends on the choice of domains~$D(\pi)$.
The domains~$D(\pi)$ are therefore determined by minimizing \cref{eq:DeltaScon}.
This procedure is exactly analogous to the random tensor network computations in \cref{sec:rtn}, if the logarithms of the bond dimensions that pass through the domain wall are replaced by geometrical area.

Recall that the boundary conditions for the permutation map are determined by the permutations applied at the asymptotic boundary (i.e., the permutations that appear in the replica computation of interest).
For R\'{e}nyi entropy computations, the boundary conditions consist of a trivial region $\bar{A}$ where no permutation is applied, and a region $A$ where a cyclic permutation is applied, as in \cref{fig:replica trick}~(a).
The dominant bulk configuration has a single domain wall lying on the minimal area bulk surface separating the two.
The R\'{e}nyi entropies are therefore equal to~$\area(\gamma_A)/4G_N$ where~$\gamma_A$ is the the minimal area surface separating~$A$ from~$\bar{A}$.

In the computation of $G(A : B : C)$ using \cref{eq:other permutations}, the boundary condition at region~$A$ is $\pi^{(1)} = (12)(34)$, at region~$B$ it is $\pi^{(2)} = (13)(24)$, and finally in region~$C$ it is the identity.
By the arguments given in \cref{prp:minimal tripartition}, the dominant permutation configuration is the minimal tripartition.
Here, a minimal tripartition for $A$, $B$, and~$C$ is a division of the bulk into three regions $\Gamma_A$, $\Gamma_B$, and~$\Gamma_C$, where the boundary intersects~$\Gamma_A$ at~$A$, and similarly for~$B$ and~$C$, and which is such that the area of the boundaries between $\Gamma_A$, $\Gamma_B$ and $\Gamma_C$ is minimal.
We map $\Gamma_A$ to $\pi^{(1)}$, $\Gamma_B$ to $\pi^{(2)}$ and $\Gamma_C$ to the identity.
The additional action of this configuration relative to the normalization factor $Z_1^4$ is $\area(\gamma_{ABC}) / 2G_N$, where $\gamma_{ABC}$ is the union of the boundaries between $\Gamma_A$, $\Gamma_B$, and~$\Gamma_C$.
Hence,
\begin{align}\label{eq:Zrepl A over 4}
    \frac{Z_\mathrm{repl}(\mathcal A)}{Z_1(\mathcal A)^4} = e^{-\frac {\mathcal A} {2G_N}},
\end{align}
where we write $\mathcal A = \area(\gamma_{ABC})$ and the notation $Z_\mathrm{repl}(\mathcal A)$ and $Z_1(\mathcal A)$ reminds us that we consider partition functions where the minimal tripartition area is fixed to be~$\mathcal A$.
We therefore find that at the leading classical order in the Newton constant~$G_N$,
\begin{align*}
    G(A:B:C) = \min_{\gamma_{ABC}}\frac {\area(\gamma_{ABC})} {4G_N}.
\end{align*}
This replica trick is illustrated in \cref{fig:replica trick}~(b).
As when we added bulk legs to the random tensor networks in \cref{sec:rtn bulk}, there is a subleading $\bigO(1)$ contribution to $G(A:B:C)$ from bulk quantum fields (if the bulk entropy is small compared to $G_N$).
If we compute the gravitational path integral using a saddle-point approximation for the geometry but by doing the full path integral over matter fields, we find that
\begin{align}
    G(A:B:C) = - \frac{1}{2} \log \left[\widehat Z_{\bulk} e^{- 2\area(\gamma_{ABC})/4G_N}\right],
\end{align}
where $\widehat Z_{\bulk}$ is the normalized partition function of the bulk matter fields on the replicated bulk geometry.
Since the gluing in the bulk partition function $\widehat Z_{\bulk}$ is the same as the gluing used to compute $G(A:B:C)$ on the boundary, with regions $A$, $B$, $C$ replaced by $\Gamma_A$, $\Gamma_B$, $\Gamma_C$ respectively, we find that
\begin{align}
   G(A:B:C) = \frac {\area(\gamma_{ABC})} {4G_N} + G_{\bulk}(\Gamma_A:\Gamma_B:\Gamma_C),
\end{align}
in close analogy to the random tensor network result in \cref{eq:rtns with bulk}.

\subsection{General semiclassical states}
What about states where the bulk spatial geometry has not been measured?
In this case every step in the derivation above goes through in the same way, until we come to actually evaluating the bulk path integral, whereupon the saddle point geometry will in general be some highly backreacted solution that obeys the Euclidean equations of motion everywhere.

Consider the simple example of the vacuum state of a 1+1-dimensional CFT on a circle, divided into three intervals, $A$, $B$, and $C$.
Each bra and ket is prepared by a boundary partition function on a hemisphere.
Gluing these partitions together in order to compute~$Z = e^{-2G}$ leads to a boundary partition function on a geometry that is topologically a sphere, but a sphere that contains six conical singularities.
The four bras and four kets each make up an octant of the sphere, with the conical singularities lying at the boundaries of the three regions, see \cref{fig:replica trick backreact}.

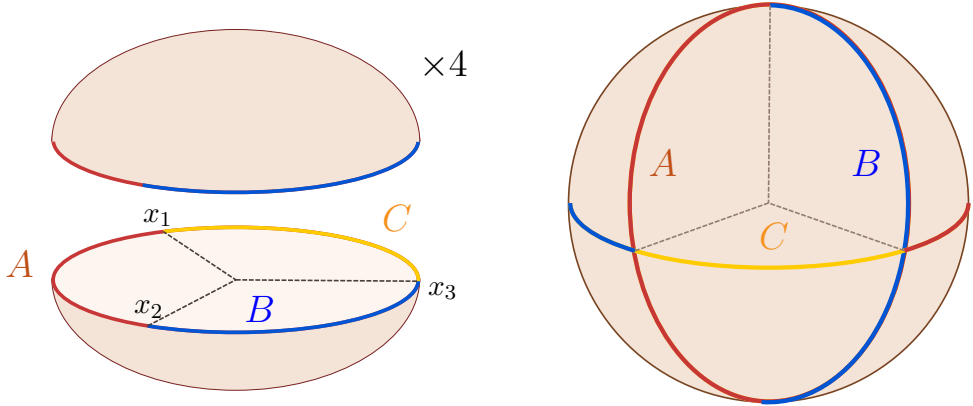
\begin{figure}
\centering
\begin{overpic}[width=.8\linewidth,grid=false]{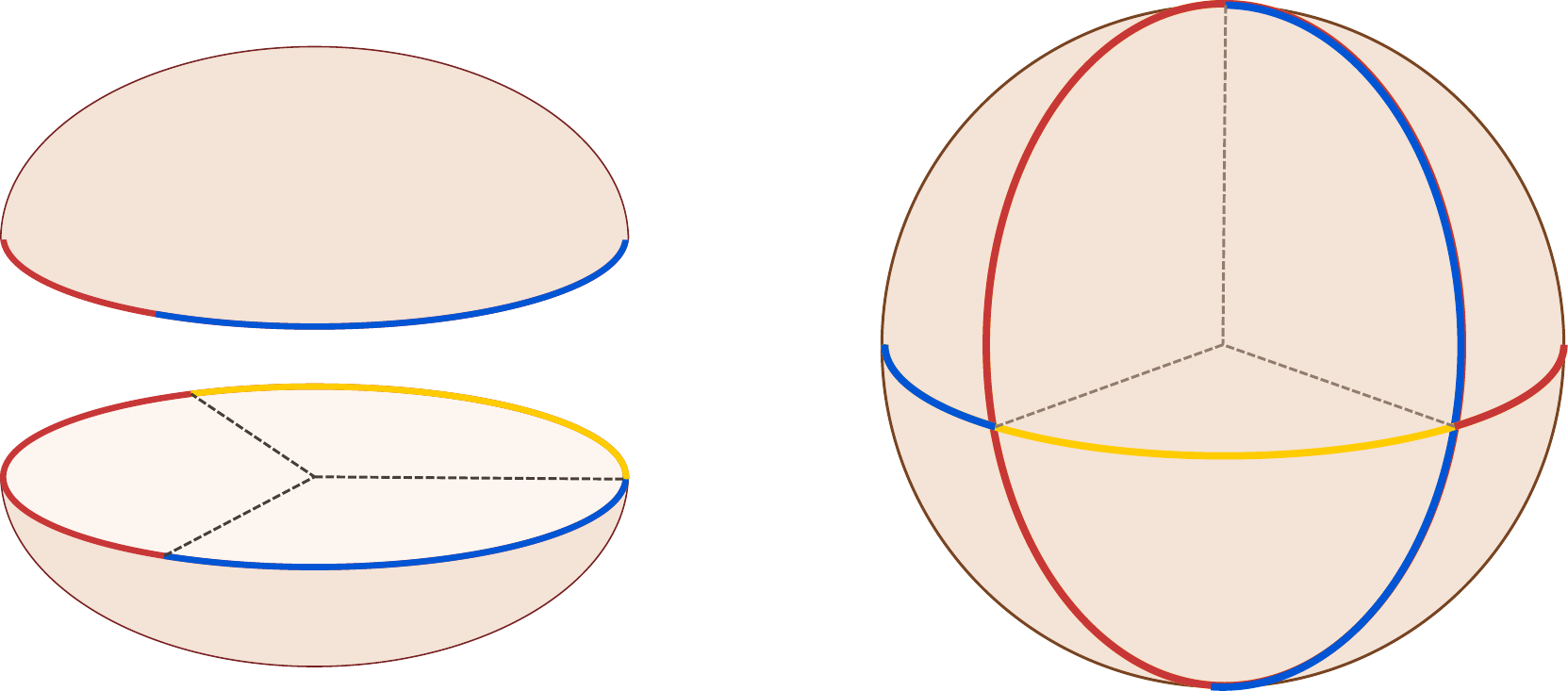}
    \put(-5,14){\color{Bittersweet}{\Large{$A$}}}
    \put(21,9){\color{Blue}{\Large{$B$}}}
    \put(36,19){\color{BurntOrange}{\Large{$C$}}}
    \put(10,20){$x_1$} \put(9,10){$x_2$} \put(41,12){$x_3$}
    \put(40,36){\Large{$\times 4$}}
    \put(65,25){\color{Bittersweet}{\Large{$A$}}}
    \put(87,25){\color{Blue}{\Large{$B$}}}
    \put(77,17){\color{BurntOrange}{\Large{$C$}}}
\end{overpic}
\caption{The replica trick for $G(A : B : C)$ where the geometry is not fixed.
The resulting backreacted geometry is a sphere, where each octant corresponds to a copy of $\ket{\Psi}$ or $\bra{\Psi}$.\label{fig:replica trick backreact}}
\end{figure}

The topology can be seen directly by gluing together bras and kets with the appropriate permutations, but it can also be verified by computing the Euler characteristic of the manifold.
Each bra and each ket contributes one face, and is bordered by three edges.
Finally, for each of the three boundaries between regions we have two vertices, because the permutations
$\pi^{(1)} (\pi^{(2)}))^{-1}$, $\pi^{(2)} (\pi^{(3)})^{-1}$, and $\pi^{(3)} (\pi^{(1)})^{-1}$
each have two cycles.

The Euler characteristic is therefore
\begin{align*}
    \chi = \text{8 faces $-$ 12 edges $+$ 6 vertices} = 2,
\end{align*}
proving that the manifold is indeed topologically a sphere.
A Weyl transformation maps the metric $g_{\mu \nu}$ of this sphere to the round metric $\hat g_{\mu \nu} = e^{-\phi} g_{\mu\nu}$.
Then the partition function is~$Z_\mathrm{repl} = e^{S_L} \widehat Z_\mathrm{sph}$, where $\widehat Z_\mathrm{sph}$ is the partition function on the sphere with the round metric and~$S_L$ is the Liouville action
\begin{align*}
    S_L = \frac{c}{96 \pi} \int \,d\hat V\,\left(\hat g^{\mu\nu} \partial_\mu \phi \partial_\nu \phi + 2 \hat R \phi\right)
\end{align*}
with $d \hat V$ and $\hat R$ respectively the volume form and Ricci scalar for the round metric $\hat g_{\mu \nu}$.
The prefactor of $e^{S_L}$ created by the rescaling of $g_{\mu \nu}$ comes from the Weyl anomaly of the CFT and hence depends only on the central charge.
The same result can also be computed from a bulk perspective using the shift in the location of the cut-off at infinity.
Because the metric $g_{\mu \nu}$ is singular at the conical singularities, the Liouville action will be divergent.
This divergence is regulated by smoothing out the conical singularity at some UV-cutoff lengthscale $\eps$.
The divergence is proportional to the conical excess and number of conical singularities, and scales as $(3 c/4)\log \eps$.
The full form of $Z_\mathrm{repl}$ is then determined entirely by the conformal symmetry and the exchange symmetry for the three regions to be
\begin{align*}
    Z_\mathrm{repl} = \frac{\eps^{3c/4}}{(x_1 - x_2)^{c/4}(x_2 - x_3)^{c/4} (x_3 - x_1)^{c/4}}.
\end{align*}
Note that any overall prefactor can be absorbed into an $\bigO(1)$ rescaling of $\eps$.
Our final result is therefore
\begin{align*}
    G (A:B:C)= \frac{c}{8} \log\left[\frac{(x_1 - x_2)(x_2 - x_3)(x_3 - x_1)}{\eps^3}\right].
\end{align*}
By comparison, the minimal area tripartition~$\gamma_{ABC}$ in vacuum AdS$_3$ has an area such that
\begin{align*}
    \frac{\area(\gamma_{ABC})}{4G_N} = \frac{c}{6} \log\left[\frac{(x_1 - x_2)(x_2 - x_3)(x_3 - x_1)}{\eps^3}\right].
\end{align*}
So the two formulas do not agree, although in this simple case they do agree up to an overall prefactor.%
\footnote{This is a consequence of conformal symmetry, analogous to the fact that R\'{e}nyi entropies for a single interval in a 2D CFT differ only by an overall prefactor.}

How can we understand the relationship of this result with our earlier computation showing that $G = \area(\gamma_{ABC})/4G_N$ in states where the bulk geometry is fixed?
The answer is fairly simple.
In the vacuum partition function computation, the classical bulk saddle point is just Euclidean AdS$_3$, which is topologically a ball with no singularities.
As shown in \cref{fig:replica trick backreact}, if we divide this ball up into eight octants, one for each bra or ket on the boundary, then each octant is bounded by a hemispherical boundary partition function (preparing the bra or ket) along with a backreacted version of the spatial slice of AdS$_3$.
Other than this backreaction, the gluing of the different octants obeys exactly the rules described in \cref{sec:fixedarea}.

Let us now write the AdS$_3$ vacuum as a superposition of states with different tripartition areas.
For each fixed value of the tripartition area $\mathcal A$, we claimed above $G(A:B:C) = \mathcal A/4G_N$.
If we assume that off-diagonal contributions can be neglected, it follows that for the vacuum state we have
\begin{align}\label{eq:Aint}
    e^{-2G} = \int d \mathcal A\,\, p(\mathcal A)^4 \,e^{-2 \frac{\mathcal A}{4G_N}}
\end{align}
where $p(\mathcal A)$ is the probability of the tripartition having area $\mathcal A$.
Generically, this will be dominated by some saddle point value for~$\mathcal A$.
For \emph{any} $\mathcal A$, $p(\mathcal A) = Z_1(\mathcal A)/Z_1$ where $Z_1$ is the full partition function on the round sphere (which computes the normalization of the vacuum state), and $Z_1(\mathcal A)$ is the same partition function, except projected onto the tripartition having area $\mathcal A$.
We already showed in \cref{eq:Zrepl A over 4} that
\begin{align*}
    \frac {Z_\mathrm{repl}(\mathcal A)} {Z_1(\mathcal A)^4} = e^{-\frac {\mathcal A} {2G_N}},
\end{align*}
where $Z_\text{repl}(\mathcal A)$ is the replicated partition function with tripartition area~$\mathcal A$.
Hence
\begin{align*}
    p(\mathcal A)^4 e^{-2 \frac{\mathcal A} {4G_N}} = \frac {Z_\mathrm{repl}(\mathcal A)} {Z_1^4}.
\end{align*}
Finally, the integrand in \cref{eq:Aint} is maximized when $\mathcal A$ is chosen so that the bulk geometry computing~$Z_\mathrm{repl}(\mathcal A)$ obeys the equations of motion everywhere (including at the tripartition), leading to a final answer $e^{-2G} = Z_\mathrm{repl}/Z_1^4$ where $Z_\mathrm{repl}$ is the replicated saddle point without fixing the tripartition area.

This is exactly what we found: $G(A:B:C)$ was computed by a bulk saddle point that looked identical to the fixed geometry computations, except that the geometry was backreacted so that the tripartition had smaller area.
Even though there was only an exponentially small probability of the vacuum having this backreacted geometry, the tripartition area was sufficiently smaller in the backreacted geometry that it dominated the computation of $G(A:B:C)$.
This story is in fact completely general, rather than being specific to three intervals in vacuum AdS$_3$.
For detailed discussions of the closely analogous situation with R\'{e}nyi entropies in fixed-area states and general states see \cite{dong2019flat, akers2019holographic}.
The only new ingredient here is the intersection of the different domain walls.
One might worry that even if we are able to smooth away all the conical singularities then the \emph{intersection} of the different domain walls may still be singular.%
\footnote{As we shall see  in \cref{sec:generalize}, this worry is very reasonable.
In the computation of generalizations of $G(A:B:C)$, the intersection of the domain walls have an inherently singular topology and cannot be smoothed out by backreacting the geometry.} Fortunately it easy to check that the boundary of a small neighbourhood of this intersection has topology $S^2 \times M$ where $M$ is the topology of the intersection.
If we smooth out the conical singularities we can fill in this neighbourhood with $B \times M$ with $B$ a smooth $3$-ball, as we saw in the simple example above.

To get a quantity that actually computes the tripartition area for the \emph{unbackreacted} geometry, in general semiclassical states, we can used the smoothed version~$G^\eps(A:B:C)$ of~$G(A:B:C)$ introduced in \cref{eq:smoothed G}.
By perturbing any semiclassical state by an $\bigO(\eps)$, we can ensure that it's geometry is fixed up to $\bigO(\sqrt{G_N} \log \eps)$ corrections~\cite{akers2021leading,marolf2020probing}.
As a result, we can lower bound $G^\eps(A:B:C)$ by the tripartition area $\area(\gamma_{ABC})/4G_N$ minus some subleading $\bigO(G_N^{-1/2} \log\frac1\eps)$ correction.
Conversely, no $\bigO(\eps)$ change to the state can stop the unbackreacted geometry from giving the dominant contribution to the wavefunction.
Since the contribution to $G(A:B:C)$ from this part of the wavefunction is determined up to an $\bigO(1)$ factor by its classical geometry, we have an upper bound on $G^\eps(A:B:C)$ given by $\area(\gamma_{ABC})/4 G_N + \bigO(1)$.
We therefore find that
\begin{align}
    G^\eps(A:B:C) = \frac{\mathrm{Area}(\gamma_{ABC})}{4 G_N} + \bigO(G_N^{-1/2} \log \eps),
\end{align}
where~$\gamma_{ABC}$ is the minimal tripartition in the unbackreacted geometry of the semiclassical state.
This closely matches the analogous result for smoothed R\'{e}nyi entropies~\cite{bao2019beyond}, which are given by $\area(\gamma_A)/4G_N + \bigO(G_N^{-1/2})$ with~$\gamma_A$ the minimal area surface homologous to~$A$.

We have focused here on states prepared by Euclidean path integrals, which always have a preferred spatial slice that is invariant under a time-reflection symmetry.
We expect that our results can be extended to general time-dependent geometries, as with entanglement entropy replica trick computations, by consider saddle points of complex geometries.
In that case the minimal tripartition (within the time-reflection symmetric slice) would be replaced by the minimal area extremal tripartition (i.e., a tripartition that has invariant area at linear order under perturbations in both time- and space-directions).
One way to construct an extremal tripartition is to uses a maximinimization procedure where one maximises, over all possible Cauchy slices, the area of the minimal tripartition within each Cauchy slice \cite{Wall:2012uf}.
However, unlike for minimal bipartitions, it is not obvious how to use a focusing argument to show that this is actually the minimal area extremal tripartition.

\section{Generalizations} \label{sec:generalize}
Given a tripartite pure state $\ket{\Psi}_{ABC}$, there is a natural family of quantities $G_n(A:B:C)_{\ket{\Psi}}$ for integer $n$ such that $G_2(A : B : C)_{\ket{\Psi}}$ coincides with $G(A : B : C)_{\ket{\Psi}}$.
We start from the reformulation in~\cref{eq:other permutations} and define permutations
\begin{align*}
  \pi^{(1)} &= \parens[\big]{ 1, 2, \dots, n } \parens[\big]{ n + 1, n+2, \dots, 2n } \cdots \parens[\big]{ n^2-n+1,n^2-n+2,\dots, n^2 }, \\
  \pi^{(2)} &= \parens[\big]{ 1, n+1, \dots, n^2-n+1 } \parens[\big]{ 2, n+2, \dots, n^2-n+2 } \cdots \parens[\big]{ n, 2n, \dots, n^2 }.
\end{align*}
That is, if one arranges the numbers $1,\dots,n^2$ in an $n\times n$ square, one row after the other, then the cycles in $\pi^{(1)}$ permute the rows, while the one in $\pi^{(2)}$ permute the columns.
Then we define
\begin{align}\label{eq:Zndef}
  Z_n(A:B:C)_{\ket\Psi} := \bra\Psi^{\ot n^2} \parens*{ \pi^{(1)}_A \ot \pi^{(2)}_B \ot \id_C } \ket\Psi^{\ot n^2}.
\end{align}
If $Z_n(A : B: C)_{\ket\Psi} > 0$ we may define
\begin{align}\label{eq:Gndef}
  G_n(A:B:C)_{\ket\Psi} := \frac1{n(1-n)} \log Z_n(A:B:C)_{\ket\Psi},
\end{align}
but we do not know if $Z_n(A:B:C)_{\ket\Psi} > 0$ in general.
This family of quantities was recently independently introduced in~\cite{gadde2022multi}.

One motivation for this particular choice of permutations from the perspective of our work is that the permutations $\pi^{(1)}$, $\pi^{(2)}$, and $\id$ are equidistant,
\begin{align*}
    d(\pi^{(1)},\pi^{(2)}) = d(\pi^{(1)},\id) = d(\pi^{(2)},\id) = n(n-1)
\end{align*}
and they are incompatible in the sense that any permutation which is on a geodesic between two of them is not on a geodesic to the third (cf.\ the discussion in \cref{subsec:incompat}).
This suggests that $G_n(A : B : C)$ may be computed by minimal tripartitions in random tensor networks and holographic settings.
Given such a family, it is natural to wonder, as was argued in \cite{gadde2022multi}, whether one can ``analytically continue'' to non-integer $n$ and obtain an entanglement measure in the $n \to 1$ limit, which, analogous to the von Neumann entropy, would be holographically dual to the minimal bulk tripartition without any need for smoothing.
Note that for random tensor networks, at least in expectation, $Z_n(A:B:C)$ is positive, as it is computed by a positive partition function.
Similarly, in holographic gravity, $Z_n(A:B:C)$ is a positive quantity if one computes the path integral semiclassically.

We will first discuss some basic properties of $G_n(A : B : C)$ and then some interesting features of holographic computations of $G_n(A : B : C)$ for larger $n > 2$ which make a naive analytic continuation analogous to that done by Lewkowycz and Maldacena \cite{lewkowycz2013generalized} impossible.

\subsection{Basic properties of \texorpdfstring{$G_n(A:B:C)$}{G\_n(A:B:C)}}
A moment's reflection shows that the quantity $Z_n(A : B : C)$ (and therefore $G_n(A:B:C)$) is invariant under permuting the subsystems $A$, $B$ and $C$ (this follows by appropriate permutations the $n^2$ bras and separately the $n^2$ kets).
It is similarly obvious that $G_n$ is additive under tensor products and invariant under local isometries, as in \cref{lem:properties G}.

It is easy to see that $Z_{n}(A:B:C)$ is real since
\begin{align}
  \overline{Z_{n}(A:B:C)}
= \tr \rho^{\ot n^2} \parens*{ \pi^{(1)}_A \ot \pi^{(2)}_B }^\dagger
= \tr \rho^{\ot n^2} \parens*{ \pi^{(1),-1}_A \ot \pi^{(2),-1}_B }
= Z_{n}(A:B:C);
\end{align}
the latter by relabeling the replicas by $i \mapsto n^2 + 1 - i$.

The quantity~$Z_{n}$ can again be interpreted in terms of a tensor network computation, generalizing the discussion in \cref{subsec:basic G}.
If we think of $\rho_{AB}$ as a PEPS tensor, with bond dimensions $d_A$ and $d_B$, then $Z_n(A:B:C)$ is given by the the PEPS contraction of $n^2$ copies of this tensor on a periodic $n \times n$ lattice, generalizing \cref{fig:gpeps}~(a).

We do not know how to prove an analogue of \cref{lem:bounds G}, but note that the argument in its proof shows that $\abs{Z_n} \leq 1$.
Thus, \emph{if} $Z_n > 0$, then $G_n(A : B : C) \geq 0$.

Let us compute~$G_n$ for the two examples previously studied for the $n=2$ case in \cref{sec:multipartite entanglement replica trick}.
In both cases $Z_n$ is positive, so $G_n$ is well-defined.
The first example is the tripartite GHZ state~\cref{eq:ghz} of local dimension~$d$.
Similarlyt to the case $n=2$, it is easy that the normalization contributes a factor~$d^{-n^2}$, whereas the trace contributes a factor~$d$ so
\begin{align*}
    \bra{\GHZ_d}^{\ot n^2} \parens*{ \pi^{(1)}_A \ot \pi^{(2)}_B \ot \id_C } \ket{\GHZ_d}^{\ot n^2}
  = d^{-n^2} \cdot d
  = d^{-(n+1)(n-1)}
\end{align*}
and hence
\begin{align*}
  G_n(A:B:C)_{\ket{\GHZ_d}}
= \frac{n+1}{n} \log d.
\end{align*}
The second example was an arbitrary state~\cref{eq:bipartite} with purely bipartite entanglement,
\begin{align*}
  \ket\Psi_{ABC} = \ket{\psi_1}_{A_1 B_1} \ot \ket{\psi_2}_{A_2 C_1} \ot \ket{\psi_3}_{B_2 C_2}.
\end{align*}
It is not hard to verify that in this case
\begin{align*}
  G_n(A:B:C)_{\ket\Psi} = \frac12 \parens*{ S_n(A)_{\ket\Psi} + S_n(B)_{\ket\Psi} + S_n(C)_{\ket\Psi} }.
\end{align*}

In these two examples, there is a natural analytical continuation of~$G_n$ to~$G_{\alpha}$ for arbitrary values of $\alpha$, in particular to $\alpha = 1$, leading to
\begin{align*}
    G_1(A:B:C)_{\ket{\GHZ_d}} &= 2\log(d) \\
    G_1(A:B:C)_{\ket{\Psi}} &= \frac12 \parens*{ S(A)_{\ket\Psi} + S(B)_{\ket\Psi} + S(C)_{\ket\Psi} }.
\end{align*}

It was assumed in \cite{gadde2022multi} that $G_n(A:B:C)$ is \emph{always} analytic in $n$, and hence that one can define a quantity~$G_n(A:B:C)$ by analytic continuation.
Of course, any such analytic continuation is necessarily non-unique, because one can always add an analytic function such as $\sin (\pi n)$ that vanishes on all integers, or a function such as $\sin (\pi n)/ (n-1)$ that vanishes only for all integer $n \geq 2$.
However, one might hope that there is a natural independent definition of~$G_\alpha(A:B:C)$ that is (a)~analytic for all real $\alpha \geq 1$ and (b)~reduces to $G_n(A:B:C)$ for integer $\alpha = n \geq 2$.
Even in the absence of such an independent definition, one might at least hope that there exists an analytic continuation~$G_\alpha(A:B:C)$ whose growth as $\abs{\alpha} \to \infty$ obeys the exponential bounds listed in Carlson's theorem (as such an analytic continuation would then be unique within the class of such analytic functions).

The idea that such an analytic continuation always exists is certainly attractive.
However, for general states we are not aware of an obvious approach to showing its existence, and we certainly do not know an independent definition of $G_1(A : B : C)$, analogous to the formula $S(\rho) = - \tr(\rho \log \rho)$ for the von Neumann entropy.

\subsection{Tensor networks and gravity}
Recall that the permutations $\pi^{(1)}$, $\pi^{(2)}$, and $\pi^{(3)}$ are equidistant and incompatible.
This makes that it is plausible that if one computes $G_{n}(A : B : C)$ for a random tensor network state the dominant configuration will still be given by a minimal tripartition.
However, a formal generalization of the proof from \cref{sec:rtn} to $G_n(A:B:C)$ for $n \geq 3$ is nontrivial, and we have not been able to find a clean and general construction for doing so.

In holographic states, on the other hand, we can show very explicitly that a replica trick saddle for~$G_n(A:B:C)$ in general does \emph{not} look like the minimal tripartition for $n \geq 3$.
Consider again the example described in \cref{sec:gravity} of vacuum AdS$_3$ with three contiguous boundary regions $A$, $B$, and $C$.
In the replica trick for $G_n(A : B: C)$ we have~$n^2$ bras and~$n^2$ kets.
The boundary topology can therefore be modelled by a polygon with $2n^2$ faces and~$3n^2$ edges.
Each of $(\pi^{(1)})^{-1}\pi^{(2)}$, $(\pi^{(2)})^{-1}\pi^{(3)}$ and $(\pi^{(3)})^{-1}\pi^{(1)}$ has $n$ cycles, so there are~$3n$ vertices and the Euler characteristic is given by
\begin{align*}
    \chi = \text{$2n^2$ faces $-$ $3n^2$ edges $+$ $3n$ vertices} = n(3-n).
\end{align*}
For instance, for $n = 3$ the boundary manifold topologically is a torus, as shown explicitly in \cref{fig:bdytorus}.

\begin{figure}
\centering
\begin{overpic}[width=0.7\textwidth,grid=false]{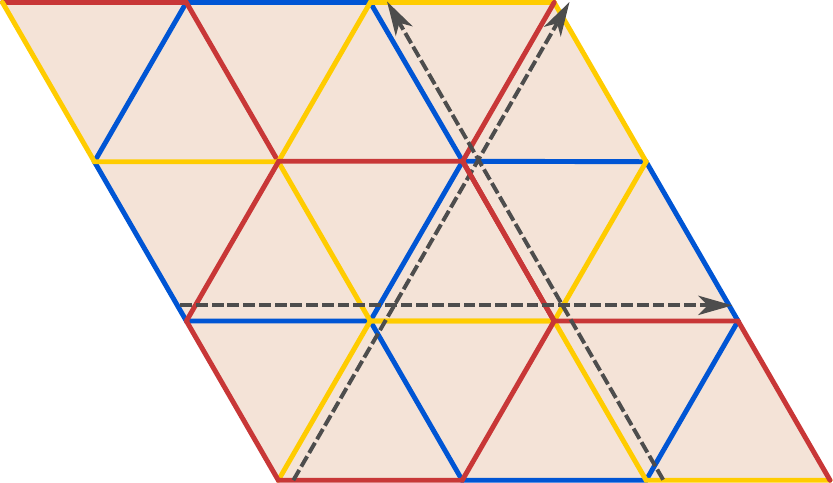}
\put(8,60){\color{Bittersweet}{\Large{$A$}}}
\put(0,46){\color{BurntOrange}{\Large{$B$}}}
\put(10,27){\color{Blue}{\Large{$C$}}}
\put(10,49){\Large{$\overline{5}$}} \put(32,49){\Large{$\overline{6}$}} \put(54,49){\Large{$\overline{3}$}}
\put(21,45){\Large{$5$}} \put(43,45){\Large{$3$}} \put(65,45){\Large{$2$}}

\put(21,30){\Large{$\overline{8}$}} \put(43,30){\Large{$\overline{1}$}} \put(65,30){\Large{$\overline{2}$}}
\put(32,26){\Large{$7$}} \put(54,26){\Large{$1$}} \put(76,26){\Large{$8$}}

\put(32,11){\Large{$\overline{7}$}} \put(54,11){\Large{$\overline{4}$}} \put(76,11){\Large{$\overline{9}$}}
\put(43,7){\Large{$4$}} \put(65,7){\Large{$6$}} \put(87,7){\Large{$9$}}
\end{overpic}
\caption{The boundary replica geometry for~$G_3$.
The $j$-th copies of path integrals computing $\ket{\Psi}$ and $\bra{\Psi}$ are labeled by $j$ and $\overline{j}$ respectively, for $j \in \{1,\dots,3^2\}$.
The resulting boundary geometry is a torus, with conformal structure labeled by $\tau = e^{i\pi/3}$, and we have an automorphism which permutes the three indicated cycles.}
\label{fig:bdytorus}
\end{figure}
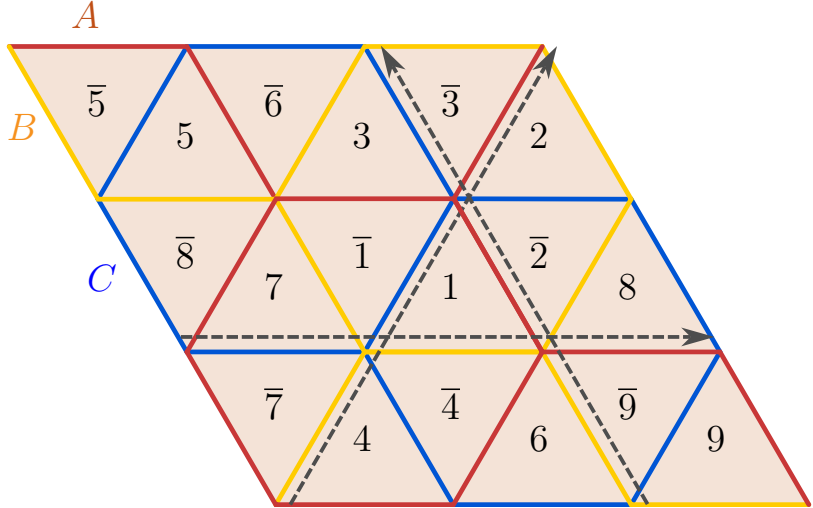

For $n = 2$, the replicated bulk geometry was a backreacted version of the minimal tripartition configuration that we found for fixed geometry states.
In particular, both were topologically a three-dimensional ball filling in the boundary two-sphere.
For $n=3$, however, the minimal tripartition configuration is topologically a cone over a torus $T^2$.
By this we mean that it is the quotient of the product space $T^2 \times [0,1]$ (with $T^2 \times \{1\}$ the conformal boundary torus) by an equivalence relation identifying $T^2 \times \{0\}$ to a single point.

Unlike a cone over $S^2$ -- which is a ball -- a cone over $T^2$ is not a topological manifold because it is singular at the origin.
The same is true for $n>3$, where the minimal tripartition configuration is topologically a cone over a higher genus surface.
As a result, the smooth geometry that dominates the replica trick computation is necessarily topologically inequivalent to the minimal tripartition configuration.
Indeed, it is well known that the dominant bulk topology for a toroidal asymptotic boundary is a solid torus $D^2 \times S^1$ where the shortest cycle of the torus (with respect to the flat metric) becomes contractible in the bulk~\cite{Witten:1998qj, Maldacena:2001kr}.
As can be seen in \cref{fig:bdytorus}, the complex structure of the $n=3$ torus has a conformal structure labeled by $\tau = e^{i \pi/3}$.
Such a torus has an enhanced $S^3$ automorphism group that permutes a set of three equal-length shortest cycles.
We therefore have a set of three dominant bulk saddles (each with equal action) where one of these three cycles becomes contractible.

We were not able to find any method of constructing such topologies by gluing together fixed-geometry bras and kets with appropriate permutations.
Since $\abs{S_9} = 9!$ is quite large, and gluing together 18 bras and kets is a nontrivial procedure to visualize, this may represent a failure of imagination on our part.
If so, it would be interesting to see whether the configuration giving the solid torus topology also dominates over the minimal tripartition for random tensor network states (perhaps with nontrivial link states), or whether backreaction, or other novel gravitational physics, is important in allowing the solid torus to dominate.
It is also possible that no such set of permutations exists, and the solid torus configuration has no analogue in random tensor networks at all.

Finally, let us comment briefly on the arguments used in~\cite{gadde2022multi} to argue that the hypothetical quantity $G_1(A:B:C)$ is dual to the minimal area tripartition.
The usual approach~\cite{lewkowycz2013generalized} to compute the von Neumann entropy using the replica trick is to use the replica trick for the $n$-th R\'enyi entropy, then realize that the boundary geometry has a $\ZZ_n$-symmetry (by permuting the copies cyclically) and assume that the bulk has the same \emph{replica symmetry}.
This then allows one to take a quotient, and the problem reduces to a geometry with a single replica, but with some conical singularity.
A key aspect of this derivation is the assumption that the bulk geometry does not break the replica symmetry.

The replica trick for $G_n(A : B : C)$ also has a natural replica symmetry given by~$\ZZ_n \times \ZZ_n$, where $(j_1,j_2) \in \ZZ_n$ maps the $(k_1 + n(k_2 - 1))$-th copy (for $k_1, k_2 = 1, \dots, n$) to the $(k_1' + n(k_2' - 1))$-th copy with $k_1' = k_1 +j_1 \pmod n$ and $k_2' = k_2 +j_2 \pmod n$.
In~\cite{gadde2022multi} the authors assume that this symmetry is not broken for the minimal bulk configurations, and use this to derive their results.
However it is easy to check that this assumption is incorrect: the minimal tripartition topology preserves this $\ZZ_3 \times \ZZ_3$ symmetry, but the three smooth solid torus geometries break it down to a $\ZZ_3$ diagonal subgroup with $j_1 = j_2$.
This can be seen explicitly in \cref{fig:bdytorus}.
The $\ZZ_3$ diagonal subgroup acts as translations of the boundary torus along the diagonal from the top left to the bottom right.
It therefore preserves the three shortest cycles.
However, all other elements of the $\ZZ_3 \times \ZZ_3$ symmetry group rotate the torus, and hence permute the three shortest cycles.
Consequently, they also permute the three bulk saddles, rather than preserving any of them, as required for the arguments in~\cite{gadde2022multi}.

\subsection*{Acknowledgements}
G.~Penington is supported by the UC Berkeley Physics Department, the Simons Foundation through the ``It from Qubit'' program, the Department of Energy via the GeoFlow consortium (QuantISED Award DE-SC0019380) and an early career award, and AFOSR award FA9550-22-1-0098.
M.~Walter acknowledges the European Research Council~(ERC) through ERC Starting Grant 101040907-SYMOPTIC, the Deutsche Forschungsgemeinschaft (DFG, German Research Foundation) under Germany's Excellence Strategy - EXC\ 2092\ CASA - 390781972, the BMBF through project QuBRA, and NWO grant OCENW.KLEIN.267.

\bibliographystyle{JHEP}
\bibliography{references}

\end{document}